\documentclass[11pt]{article}

\usepackage{natbib}                         
\usepackage{amsmath, amsfonts, amssymb} 
\usepackage[only,longmapsfrom]{stmaryrd}   
\usepackage{comment}
\usepackage{graphicx}

\usepackage{amsthm}            

\newtheorem*{theorem*}{Theorem}
\newtheorem{lemma}{Lemma}
\newtheorem{prop}{Proposition}

\usepackage{color}             
\usepackage[T1]{fontenc}       
\usepackage{lmodern}
\usepackage{bbm}               

\usepackage{pdflscape}         
\usepackage{graphicx}          

\usepackage{arydshln}          
\usepackage{accents}           
\usepackage{blkarray}          

\usepackage{tikz-cd}           

\usepackage{algorithm}         
\usepackage{algorithmic}

\usepackage{setspace}          

\newcommand{\bd}{\boldsymbol}

\newcommand{\defi}{\stackrel{\text{def}}{=}}	
\newcommand{\first}[2]{  \ensuremath{ \frac{\partial   #1}{\partial #2  } }  }

\newcommand{\redud}[1]{\textcolor{red}{ \underline{\textbf{#1}} }}

\DeclareMathAccent{\wtilde}{\mathord}{largesymbols}{"65}   
\newcommand{\ut}[1]{\underaccent{\wtilde}{#1}}             

\def \spacingset#1{ \renewcommand{\baselinestretch}{#1}\small\normalsize } \spacingset{1}

\begin{document}


\title{``Stochastic Inverse Problems'' and Changes-of-Variables
   \thanks{This work was supported in part by the NNSA Office of Defense Nuclear Nonproliferation Research and Development, NA-22 F2019 Nuclear Forensics Venture (LA19-V-FY2019-NDD3Ac).}
}

\author{
Peter W. Marcy\thanks{Los Alamos National Laboratory, Los Alamos, NM ({pmarcy@lanl.gov})}
\and  
Rebecca E. Morrison\thanks{University of Colorado Boulder, Boulder, CO ({rebeccam@colorado.edu})}
} 

\date{}

\maketitle

\begin{abstract}
\noindent
Over the last decade, a series of applied mathematics papers have explored a type of inverse problem---called by a variety of names including ``inverse sensitivity'', ``pushforward based inference'', ``consistent Bayesian inference'', or ``data-consistent inversion''---wherein a solution is a probability density whose pushforward takes a given form.
The formulation of such a \emph{stochastic} inverse problem can be unexpected or confusing to those familiar with traditional Bayesian or otherwise statistical inference.  
To date, two classes of solutions have been proposed, and these have only been justified through applications of measure theory and its disintegration theorem.
In this work we show that, under mild assumptions, the formulation of and solution to \emph{all} stochastic inverse problems can be more clearly understood using basic probability theory: 
a stochastic inverse problem is simply a change-of-variables or approximation thereof.
For the two existing classes of solutions, we derive the relationship to change(s)-of-variables and illustrate using analytic examples where none had previously existed.
Our derivations use neither Bayes' theorem nor the disintegration theorem explicitly.
Our final contribution is a careful comparison of changes-of-variables to more traditional statistical inference.
While taking stochastic inverse problems at face value for the majority of the paper, our final comparative discussion gives a critique of the framework.  \\

\noindent
Keywords:  Jacobian, reparameterization, uncertainty quantification, statistical inference, Bayesian analysis
\end{abstract}


\section{Introduction}

Scientific models often relate unknown parameters or functions to observable quantities.
Any choice of unknown input to the model will produce an observable that can be compared with real measurements.
In the absence of noise or measurement error, this \emph{forward} problem can be thought of as a well-defined operator from parameter/function-space to observable-space.
To solve the \emph{inverse} problem (IP) is to estimate the unknown inputs from a finite collection of observations.
In the more formal language of \cite{Tikhonov1995}, if $\mathcal F$ is an operator between metric (or Banach, or Hilbert) spaces representing parameters and data, $\mathcal F: \mathcal P \mapsto \mathcal Q$,
the solution to the IP is defined by the operator equation
\begin{align}
  \bd \theta
  &\text{\quad such that \quad}
  \mathcal F(\bd\theta) = \bd y
  \quad\quad  \big( \bd\theta \in \mathcal P, \  \bd y \in \mathcal Q \big)  \ .   \label{eq:InvProbClassic}
\end{align}
This operator formulation is typically associated with function spaces, where solving the IP means solving equations (often integral equations) given a finite number of potentially noisy observations \citep{Groetsch1993, Stuart2010, Kirsch2011, Aster2019, Lesnic2021}.  
Equality within (\ref{eq:InvProbClassic}) may not be possible, and a solution might instead be required to minimize some functional involving $\mathcal F(\bd\theta)$ and $\bd y$, as in the case of weighted or regularized least-squares estimation.

When $\mathcal P$ and $\mathcal Q$ are subsets of $\mathbb R^p$ and $\mathbb R^q$ (respectively), the forward operator is a Euclidean map $\bd g$, and the noise is typically written into the IP explicitly.
The solution is
\begin{align}
  \bd \theta
  &\text{\quad such that \quad}
  \bd g(\bd\theta) + \bd\epsilon = \bd y 
  \quad\quad  \big( \bd\theta \in \mathcal P, \  \bd y \in \mathcal Q, \  \bd\epsilon \sim F_{\bd E} \big)  \ .   \label{eq:InvProbClassicNoise}
\end{align}
The term $\bd\epsilon$ is a realization of a random variable $\bd E$ corresponding to the error process with cumulative distribution function $F_{\bd E}$.
(The use of the cumulative function allows for continuous as well as discrete and mixed random variables.
Also, the statement within (\ref{eq:InvProbClassicNoise}) assumes additive error, but is trivially modified for multiplicative or other error structures.)
Together, $\bd g(\bd\theta)$ and $F_{\bd E}(\bd\epsilon)$ determine a model whose likelihood function forms the basis for \emph{statistical} IPs \citep{Pawitan2001, Davison2003, Kaipio2005, Tarantola2005, Stuart2010, Reid2010, Chiachio2022}.
As stated in (\ref{eq:InvProbClassicNoise}), the solution to the statistical IP is a point estimate derived in some fashion from the likelihood.
In certain cases, the maximum likelihood estimate will be a weighted or regularized least-squares estimate.
In the Bayesian \citep{Bernardo1994, Robert2007, Gelman2014} and Fiducial \citep{Hannig2009, Hannig2016} inferential contexts, the solution to a statistical IP is an entire probability distribution from which point estimates may be obtained.

Over the past century, the field of inverse problems has extended into nearly every domain of science, engineering, and technology, cementing its status as fundamental to applied mathematics and statistics.
One particular, more recently developed class of so-called \emph{stochastic} inverse problems (SIPs) identifies a solution as a probability density whose pushforward takes a given form 
\citep{Butler2014, Butler2015, ButlerHuhtala2015, Mattis2015, Butler2018a, Mattis2019, Uy2019, Butler2020a, Bruder2020, Butler2020d, Tran2021}.
These SIPs and their solutions have also gone under the names of
``(stochastic) inverse sensitivity problems'' \citep{Breidt2011, Butler2012, Butler2014, Butler2013, Graham2017};
``measure-theoretic inverse problems'' \citep{Butler2017, Presho2017}; 
``consistent Bayesian'' or ``pushforward-based inference'' \citep{Butler2018a, Butler2018b, Walsh2018};
``data/observation-consistent inversion'' \citep{Butler2018b, Butler2020b, Butler2020c, Mattis2022}; or
``random parameter models'' \citep{Swigon2019}.

The formulation and solution of these SIPs can look peculiar to those familiar with more traditional inverse problems. 
For example, each observable quantity is an entire probability density, not of a fixed realization of a random variable modelled conditionally.
Also, there must be at least as many unknown parameters as observables.
Moreover, derivations of the two different classes of solutions that have been proposed---those of \cite{Breidt2011}  and \cite{Butler2018a} which we call ``BBE'' (short for ``Breidt, Butler, Estep'') and ``BJW'' (short for ``Butler, Jakeman, Wildey''), respectively---rely heavily upon measure theory, specifically the disintegration theorem.

The goal of this paper is to explore the formulation and solution to SIPs using only introductory probability and mathematical statistics while also giving a careful examination of the existing literature.
While measure theory is often the language of choice, we feel that in this context it can obscure concepts that are quite uncomplicated.
To begin to explore some of the SIP peculiarities, we first address the case when $p = q$.
For $p \geq q$, we then offer a class of ``intuitive'' solutions and show how these relate to a change-of-variables (CoV) from observable- to parameter-space (Section~\ref{sec:Int}).
We give a simple algorithm to obtain samples from intuitive solutions, and the underlying reasoning proves useful for later results.
The BBE and BJW solutions are investigated through theoretical derivation and analytic examples (Sections~\ref{sec:BBE},~\ref{sec:BJW}).
Analytic results have heretofore never been given.
We show that under mild assumptions, the BBE and BJW solutions can be derived from CoVs.
Some discussion of related work is provided in Section \ref{sec:OtherWork}.
We then show that \emph{any} solution to an SIP must be related to a CoV (Section~\ref{sec:ALLisCov}).
For $p>q$, the results rely on \emph{auxiliary variables} to augment the underdetermined system.
A comparative critique of SIPs/CoVs versus inference is given in Section~\ref{sec:Inference} along with two illustrative examples.
Our main conclusion is that SIPs are significantly different from statistical inverse problems and should be handled with greater care.

In the remainder of this section, we give a formal definition of the SIP 
and a statement of the traditional CoV theorem. 

\subsection{Stochastic Inverse Problem Formulation and Assumptions}   \label{sec:SIP}

In an SIP, a \\ $q$-dimensional vector of observable quantities or data is taken to be a random variable $\bd Y$ with given probability density function (pdf) $f_{\bd Y}(\bd y)$. 
Further, there is also a function or ``forward-map'' from parameter- to data-space
\begin{align*} 
  \bd g : \big( \mathcal P \subseteq \mathbb R^p \big)  \ &\longmapsto \  \big( \mathcal Q \subseteq \mathbb R^q \big)  \\
  \left[  \begin{array}{l}
  \Theta_1  \\
  \vdots  \\ 
  \vdots  \\ 
  \Theta_p
  \end{array}  \right] 
  &\longmapsto
  \left[  \begin{array}{lcl}
  Y_1  &  \hspace*{-6pt} = \hspace*{-6pt}  &  g_1(\bd\Theta)  \\
  \vdots  & &  \\ 
  Y_q  &  \hspace*{-6pt} = \hspace*{-6pt}  &  g_q(\bd\Theta)  
  \end{array}  \right] , 
\end{align*}
with $p \geq q$,
which is either known analytically or which can be evaluated as a ``black-box'', such as a computer model that solves a set of differential equations.
The goal is to obtain a density for random variables in the pre-image, $\bd\Theta \in \mathcal P$, that will transform (exactly or in some approximate sense) to $f_{\bd Y}(\bd y)$ under the forward-map $\bd g$.
In other words, given $f_{\bd Y}$ and $\bd g$, the solution is 
\begin{align}
  f_{\bd\Theta}(\bd \theta)
  &\text{\quad such that \quad}
  \bd \Gamma  \defi  \bd g(\bd\Theta) \sim f_{\bd Y}   \label{eq:InvProb}
\end{align}
i.e., a density that pushes forward or propagates ``correctly''. 
\emph{Throughout this paper we assume that a solution exists}, though in applications one may need to be careful and check that the range $\bd g(\mathcal P)$ contains $\mathcal Q$, the support of the given $f_{\bd Y}(\bd y)$.

Within the classical IP definition (\ref{eq:InvProbClassic}), taking $\mathcal P$ and $\mathcal Q$ to be subsets of Euclidean space, and replacing $\bd y$ and $\bd\theta$ with continuous random variables $\bd Y$ and $\bd\Theta$, then an SIP defined by (\ref{eq:InvProb}) appears to be a natural variant of the more traditional (\ref{eq:InvProbClassic}).
The simplest example is when the operator $\mathcal F$ is a linear map $\bd g$ between Euclidean spaces and represented by the invertible matrix $\bd A$. 
The solution to the classical linear IP is of course $\bd\theta = \bd A^{-1} \bd y$; the solution to the linear SIP is the density of the random variable $\bd A^{-1} \bd Y$. 

In practice, a random sample from $f_{\bd\Theta}(\bd \theta)$ constitutes a solution as well, albeit an \emph{approximate} solution whose quality increases with sample size.  
Let us explicitly state this as a non-controversial assumption.

\noindent 
\emph{
\textbf{\underline{Assumption (A0):}}  
The SIP defined by (\ref{eq:InvProb}) is approximately solved when a random sample 
$\bd\theta^{(1)}, \ldots, \bd\theta^{(M)}$ is obtained such that 
$\bd g \big(\bd\theta^{(1)}\big), \ldots \bd g\big(\bd\theta^{(M)}\big) \stackrel{\text{iid}}{\sim} f_{\bd Y}$.
}

The main assumption that we use in this paper is the following.

\noindent 
\emph{
\textbf{\underline{Assumption (A1):}}  The maps $g_1, \ldots, g_q$ are in $C^1(\mathcal P)$ (i.e. they are continuously differentiable on the domain) and the Jacobian $\big| \first{\bd g}{\bd\theta} \big|$ has full row-rank except for on a set $\mathcal P_0$ of measure zero.  
Without loss of generality, the left $q \times q$ block of the Jacobian $\big| \first{\bd g}{\bd\theta_{1:q}} \big|$ is invertible on all of $\mathcal P \setminus \mathcal P_0$.
}

First of all, it is not unreasonable to assume the Jacobian has linearly independent rows almost everywhere (``a.e.'', meaning for all but a set of measure zero) in $\mathcal P$.
If the rows were dependent a.e., then one could consider only the independent outputs and instead solve the corresponding reduced SIP.
Second, (A1) allows the domain to be partitioned into subdomains $\mathcal P_0$ where the Jacobian determinant vanishes and disjoint sets $\mathcal P_1, \ldots, \mathcal P_m$ where the functions $g_1, \ldots, g_q$ produce full row-rank Jacobian.

In the first paragraph of this section it was taken as given that: 

\noindent 
\emph{
\textbf{\underline{Assumption (A2):}}  The random variables 
$\bd Y$ and $\bd\Theta$ are absolutely continuous and thus admit probability densities with respect to Lebesgue measure.
}

This assumption is not strictly necessary but is, nonetheless, the usual context for the SIPs.
The results of this paper carry over to discrete random variables as well, but we shall from here on only speak in terms of probability densities.

\subsection{Changes-of-Variables}\label{sec:CoV}

For completeness we state the change-of-variables (CoV) theorem,
which can be found in standard textbooks on probability and mathematical statistics. The following is taken from \cite{Shao2003}, pg. 23 and is also similar to \cite{Casella2002}, pg. 185.
\begin{theorem*}[Change-of-Variables] 
Let $\bd U$ be a random $p$-vector with a Lebesgue pdf $f_{\bd U}(\bd u)$ and let $\bd V \defi \bd t(\bd U)$, where $\bd t$ is a Borel function from $(\mathbb R^p, \mathcal B^p)$ to $(\mathbb R^p, \mathcal B^p)$. 
Let $A_1, \ldots, A_d$ be disjoint sets in $\mathcal B^p$ such that $\mathbb R^p - (A_1 \cup \cdots \cup A_d)$ has Lebesgue measure 0 and $\bd t$ on $A_i$ is one-to-one with a nonvanishing Jacobian, i.e., $| \first{\bd t}{\bd u} | \neq 0$ on $A_i$ for $i=1, \ldots, d$.  
Then $\bd V$ has the following Lebesgue pdf
\begin{align}
   f_{\bd V}(\bd v)  &=  \sum_{i=1}^d  f_{\bd U}\big( \bd u = \bd t^{-1}_i(\bd v) \big)  \Big| \first{\bd t^{-1}_i}{\bd v} \Big|  \ ,   \label{eq:CoV}
\end{align}
where $\bd t^{-1}_i$ is the inverse function of $\bd t$ on $A_i$ (i=1, \ldots, d).
\end{theorem*}
Being that one of the goals of the paper is to keep exposition at the level of basic probability and away from measure theory, we note that the theorem above will be the only place in which Borel sigma algebras are mentioned, and the last time we need to refer to Lebesgue measure.

\section{CoV Solutions to the SIP When $p = q$: Immediate Issues and Insights}   \label{sec:CoV_SIP}   

Taking $\bd t = \bd g$ and $\bd U = \bd\Theta$ in the CoV Theorem, the pushforward of $f_{\bd\Theta}$ through $\bd g$ has density 
$f_{\bd\Gamma}(\bd\gamma)  =  \sum_{i=1}^d  f_{\bd\Theta}\big( \bd g_i^{-1}(\bd\gamma) \big)  \big| \partial\bd g_i^{-1} / \partial\bd\gamma \big|$.
If $\bd g$ is a one-to-one function of $\mathcal P$ onto $\mathcal Q$, then $d = 1$, and no $i$ subscript is necessary.
In this case one can go the other direction and pull $f_{\bd Y}$ back through $\bd h \defi \bd g^{-1}$ to get 
$f_{\bd H}(\bd\eta)  =  f_{\bd Y}\big( \bd g(\bd\eta) \big)  \big| \partial\bd g / \partial\bd\eta \big|$.
Thus, by taking $\bd\Theta \defi \bd H$, the SIP is solved uniquely (a.e.) by the density
\begin{align}
   f^{\text{CoV}}_{\bd\Theta}(\bd\theta)  &=  f_{\bd Y}\big( \bd g(\bd\theta) \big)   \Big| \first{\bd g}{\bd\theta} \Big|  \ .  \label{eq:CoVsolution}
\end{align}

If $\bd g$ is not a uniquely invertible function, its many-to-one nature allows for multiple solutions to the SIP.
We show in the next proposition that there are in fact \emph{infinitely} many solutions.
This property does not appear to have been observed within the existing literature.
For simplicity and ease of exposition we assume that $\bd g$ is exactly $m$-to-1 onto all of $\mathcal Q$.
The general case can be proved similarly, but involves keeping track of all the distinct sets where the function is two-to-one, three-to-one, etc{.}~within the $d$-partition of the domain; as such the indexing is more laborious and no further insight is added. 
\begin{prop} 
Assume (A1) and let $\bd g$ be an $m$-to-1 ($m > 1$) function of $\mathcal P$ onto $\mathcal Q$.  
Then there exists a continuously-indexed, infinite family of solutions to the SIP.
\label{prop:InfiniteSolutions}
\end{prop}
\begin{proof}
Let $\mathbbm{1} \{ \bd\theta \in S \}$ denote an indicator function that is 1 when $\bd\theta$ is in the set $S$, and 0 otherwise.

Consider the density 
\begin{align}
   f^{\text{CoV}}_{\bd\Theta,\bd w}(\bd\theta)  &\defi
   f_{\bd Y}\big( \bd g(\bd\theta) \big)  \Big| \first{\bd g}{\bd\theta} \Big|  
   \sum_{j = 1}^m  w_j  \  \mathbbm 1\{ \bd\theta \in \mathcal P_j \}
\end{align}
continuously indexed by mixture weights summing to unity: $\sum_{j = 1}^m w_j = 1$.
By the CoV Theorem, the pushforward density for $\bd\Gamma \defi \bd g(\bd\Theta)$ is
\begin{align*}
   f_{\bd\Gamma}(\bd\gamma)  
   &=  \sum_{i = 1}^m  f^{\text{CoV}}_{\bd\Theta,\bd w} \big( \bd\theta = \bd g_i^{-1}(\bd\gamma) \big)  \  \Big| \first{\bd g_i^{-1}}{\bd\gamma} \Big|  \\
   &=  \sum_{i = 1}^m \left( f_{\bd Y}\big( \bd g(\bd g_i^{-1}(\bd\gamma) \big)  \Big| \first{\bd g}{\bd\theta} \Big|_{\bd\theta = \bd g_i^{-1}(\bd\gamma)}  \sum_{j = 1}^m  w_j  \  \mathbbm 1\{ \bd g_i^{-1}(\bd\gamma) \in \mathcal P_j \}  \right)  \Big| \first{\bd g_i^{-1}}{\bd\gamma} \Big|  \\
   &= f_{\bd Y}(\bd \gamma)  \sum_{i = 1}^m \left( \sum_{j = 1}^m  w_j  \  \mathbbm 1\{ \bd g_i^{-1}(\bd\gamma) \in \mathcal P_j \}  \right)  \\
   &= f_{\bd Y}(\bd \gamma)  \sum_{i = 1}^m  w_i  \  \mathbbm 1\{ \bd g_i^{-1}(\bd\gamma) \in \mathcal P_i \}  
   \quad =  f_{\bd Y}(\bd \gamma)
\end{align*}
as desired. 
In the penultimate line, all of the $i \neq j$ terms vanish tautologically.
\end{proof}
Section \ref{ssec:two-to-one} of the appendix gives a concrete demonstration of the method used in the proof above.
In the general case, any sets in the domain that get mapped to the same range subset can be given their own set of weights summing to unity.
%
In summary, when $p = q$, the existence of a neighborhood within $\mathcal P$ where $\bd g$ is not one-to-one implies an infinite number of solutions to the SIP.
The following related result states that there are no other solutions to the SIP beyond those found by CoVs.

\begin{prop} 
Suppose that $p=q$ and assumption (A1) holds for the forward map $\bd g$.
If $f_{\bd\Theta}(\bd\theta)$ solves the SIP (\ref{eq:InvProb}), then it is (a.e.) derivable from changes-of-variables.
\label{prop:CoVpq}
\end{prop}
\begin{proof}
By (A1), the domain $\mathcal P \setminus \mathcal P_0$ can be partitioned into $\mathcal P_1, \ldots, \mathcal P_d$ such that $\bd g$ has $C^1(\mathcal Q)$ inverses $\bd g_1^{-1}, \ldots, \bd g_d^{-1}$ defined on these sets.
The density $f_{\bd\Theta}(\bd\theta)$ solves the SIP, so then by the CoV Theorem, its pushforward under $\bd g$ is 
$f_{\bd Y}(\bd y)  =  \sum_{i=1}^d  f_{\bd\Theta}\big( \bd g_i^{-1}(\bd y) \big)  \ \big| \first{\bd g_i^{-1}}{\bd y} \big|$.
Let the range sets be denoted $\mathcal Q_i \defi \bd g(\mathcal P_i)$ (while not a proper partition, it still holds that $\mathcal Q = \bigcup_{i=1}^d \mathcal Q_i$).
Taking
\begin{align*}
   f_{\bd Y_i}(\bd y)  &\defi  w_i^{-1}  f_{\bd\Theta}\big( \bd g_i^{-1}(\bd y) \big)  \ \Big| \first{\bd g_i^{-1}}{\bd y} \Big| \  \mathbbm 1\{ \bd y \in \mathcal Q_i \}  \\
   w_i  &=  \int_{\mathcal Q_i} f_{\bd\Theta}\big( \bd g_i^{-1}(\bd y) \big)  \ \Big| \first{\bd g_i^{-1}}{\bd y} \Big| d\bd y  
   \quad = \int_{\mathcal P_i} f_{\bd\Theta}(\bd\theta)  d\bd\theta  \ ,
\end{align*}
a CoV from each $\mathcal Q_i$ to $\mathcal P_i$ under $\bd g_i^{-1}$ yields densities $f^{\text{CoV}}_{\bd\Theta_i}(\bd\theta) \defi w_i^{-1}  f_{\bd\Theta}(\bd\theta) \mathbbm 1\{ \bd\theta \in \mathcal P_i \}$.
The weighted combination $\sum_{i=1}^d w_i \ f^{\text{CoV}}_{\bd\Theta_i}(\bd\theta)$ of the $d$ range subset CoVs is then the original given solution to the SIP, up to a set of exceptional points having measure zero.
\end{proof}

Regardless of whether $\bd g$ is invertible everywhere, if the pushforward of $f_{\bd\Theta}$ is $f_{\bd Y}$, then the original density can be written concisely as a CoV from $\mathcal Q$ to $\mathcal P \setminus \mathcal P_0$.
To see this note that by the Inverse Function Theorem, 
$\big| \first{\bd g_i^{-1}}{\bd y} \big|_{\bd y = \bd g(\bd\theta)}  =  \big| \first{\bd g}{\bd\theta} \big|^{-1}$ for all $i$ so that
\begin{align}
   f_{\bd Y}(\bd y)  &=  \sum_{i=1}^d  f_{\bd\Theta} \big( \bd g_i^{-1}(\bd y) \big)  \ \Big| \first{\bd g_i^{-1}}{\bd y} \Big|  \nonumber \\
   \Rightarrow \quad
   f_{\bd Y}(\bd g(\bd\theta))  &=  
   \Big| \first{\bd g}{\bd\theta} \Big|^{-1} \ 
   \sum_{i=1}^d  f_{\bd\Theta}\big( \bd g_i^{-1}(\bd g(\bd\theta)) \big)  \nonumber  \\
   &=  \Big| \first{\bd g}{\bd\theta} \Big|^{-1} \
   \sum_{i=1}^d  f_{\bd\Theta}(\bd\theta) \ \mathbbm 1\{ \bd\theta \in \mathcal P_i \}  \nonumber  \\
   \Rightarrow \quad
   f_{\bd\Theta}(\bd\theta)  &=  f_{\bd Y}(\bd g(\bd\theta)) \ \Big| \first{\bd g}{\bd\theta} \Big|  \ .  \label{eq:PushforwardPullback}
\end{align}
This is a general statement that a density can be written in terms of its pushforward, and will be used in Proposition \ref{prop:BJWisCoV}.

\section{Intuitive Solutions to SIPs}   \label{sec:Int}

It was demonstrated in the last section that when $p = q$, the SIP (\ref{eq:InvProb}) is solved by using at least one density given by the CoV theorem.
However, for the more general case of $p > q$, there is also a family of densities that solve the SIP (\ref{eq:InvProb}) which we call ``intuitive solutions''.

Let the parameter space be partitioned as $\bd\Theta^\top = \big( \bd\Theta_{1:q}^\top, \bd\Theta_{(q+1):p}^\top \big)$.
Temporarily fix \\ 
$\bd\Theta_{(q+1):p} = \bd\theta^*_{(q+1):p}$ and consider the distribution that comes from solving the ``square'' SIP: 
\begin{align}
   f_{\bd\Theta_{1:q} | (\bd\Theta_{(q+1):p} = \bd\theta^*_{(q+1):p})}^\text{CoV}
   &\text{\quad such that \quad}
   \bd g\big( \bd\Theta_{1:q} , \bd\theta^*_{(q+1):p} \big)  \sim  f_{\bd Y}
   \label{eq:Intuitive1}
\end{align}
via a $q$-dimensional change-of-variables.
This implies that the intuitive solution $f_{\bd\Theta}^{\text{Int}}$ given by
\begin{align}
   f_{\bd\Theta}^{\text{Int}}  &\defi
   f_{\bd\Theta_{1:q} | \bd\Theta_{(q+1):p}}^\text{CoV}  \cdot  f_{\bd\Theta_{(q+1):p}}  \label{eq:Intuitive2}
\end{align}
solves the SIP for any choice of $f_{\bd\Theta_{(q+1):p}}$.

Sampling from this density constitutes a solution by (A0), and this can be performed using a simple and intuitive (hence the name) Monte Carlo routine within the algorithm below.
\begin{algorithm}[H]
\caption{Intuitive solutions to the stochastic inverse problem (SIP)}   \label{alg:Intuitive}  
\vspace*{0.1in}
Given densities $f_{\bd Y}$ and $f_{\bd\Theta_{(q+1):p}}$,
\begin{algorithmic}
   \item[$\langle 1 \rangle$]  Sample $\bd y^* \sim f_{\bd Y}$ and $\bd\theta_{(q+1):p}^* \sim f_{\bd\Theta_{(q+1):p}}$.
   \item[$\langle 2 \rangle$]  Solve the (likely nonlinear) equation  
   \begin{align} 
     \bd y^* - \bd g\big( \bd\theta_{1:q}^*, \bd\theta_{(q+1):p}^* \big)  = \bd 0   
   \end{align}
   for $\bd\theta_{1:q}^*$ to form the random sample $\big( \bd\theta_{1:q}^{*\top}, \bd\theta_{(q+1):p}^{*\top} \big)^\top$, a realization of $\bd\Theta \sim f^{\text{Int}}_{\bd\Theta}$.
   \item[$\langle 3 \rangle$]  Repeat $\langle 1 \rangle - \langle 2 \rangle $.
\end{algorithmic}
\end{algorithm}
Algorithm \ref{alg:Intuitive} requires no sophisticated sampling techniques, but it does require a nonlinear solver.
Typical solvers benefit from the ability to evaluate and use the Jacobian, but this is not always required.
This algorithm was mentioned in \cite{Swigon2019} for the simplest case of $p=q$,  
though the authors ultimately avoided the use of solvers and instead employed a Metropolis-Hastings algorithm based upon an approximated Jacobian term.  
Intuitive solutions as given above were not considered as an alternative within the work stemming from the two major approaches of BBE \citep{Breidt2011} and BJW \citep{Butler2018a}.

When the Jacobian of the function $\bd g$ vanishes within the domain $\mathcal P$, the input space must be partitioned into $d > 1$ disjoint sets according to the CoV Theorem.
This will lead to the same infinite class of solutions as in Proposition \ref{prop:InfiniteSolutions} of the previous section.
When using deterministic equation solvers, the distribution of starting values for $\bd\theta^*_{1:q}$ will determine the particular density out of the infinite family.

The main characteristic of intuitive solutions is made explicitly clear in Algorithm \ref{alg:Intuitive}.
In the first step, $\bd y^* \sim f_{\bd Y}$ and $\bd\theta_{(q+1):p}^* \sim f_{\bd\Theta_{(q+1):p}}$ are drawn \emph{independently}, not from some joint density.
This implies that the solution has the feature that the response and $(p-q)$ of the parameters do not covary---the output is independent of these inputs!
Thus, when forming an intuitive solution to the SIP, one might specify a density for $(p-q)$ of the least important input parameters, as determined by an initial sensitivity analysis.
%

As a final note about intuitive solutions, the relationship between the density (\ref{eq:Intuitive2}) and the sampling steps of Algorithm \ref{alg:Intuitive} will prove useful in later sections (specifically, \ref{sec:BBEisCOV}, \ref{sec:BJWisCOV}, and~\ref{sec:ALLisCov}).

\section{ Breidt et al. 2011 (BBE) and Related Work}   \label{sec:BBE}

The IP under current consideration was first (to the best of our knowledge) described in \cite{Breidt2011} under the name ``inverse sensitivity problem'' in the first of three papers (Parts I, II, III). 
The authors give an algorithm to compute an approximate pdf of the solution, which we denote as $\widehat f^{\text{BBE}}_{\bd\Theta}$ 
for a single output ($q=1$). However, the authors do not specify the exact solution $f^{\text{BBE}}_{\bd\Theta}$ that should be approximated by $\widehat f^{\text{BBE}}_{\bd\Theta}$. 
This approximate solution relies upon (A1), (A2), and the presence of derivative information, obtained either analytically or estimated via adjoint techniques.
The algorithm and hence the approximate solution $\widehat f^{\text{BBE}}_{\bd\Theta}$ depends heavily upon discretization and as such is restricted to very low dimension $p$.

Part II \citep{Butler2012} gives a rigorous error analysis of \cite{Breidt2011}, treating both statistical error due to sampling of the data density $f_{\bd Y}$, and numerical error due to solving differential equations during the likelihood calculation.
Part III \citep{Butler2014} considers the multiresponse scenario, i.e., $q >1$. 
The Part III approach deviates from Parts I and II in that the authors discretize events in the parameter space as opposed to manifolds in the observable space. 
This framework leads to the first mention of the disintegration theorem of measure theory.

The last methodological developments within the BBE framework are \cite{Butler2017} and \cite{Mattis2019}.
\cite{Butler2017} discusses the necessary event approximations and considers an adaptive sampling algorithm to solve the SIP.
\cite{Mattis2019} focuses on computational estimates of events from samples, using adjoints to enhance low-order piecewise-defined surrogate response surfaces.

Several subsequent works have explored the use of $\widehat f^{\text{BBE}}_{\bd\Theta}$ in various contexts. 
\cite{Butler2013} illustrates the BBE approach using the Brusselator reaction–diffusion model.
\cite{ButlerHuhtala2015} compares the multiresponse methodology of \cite{Butler2014} to other UQ methods in the context of material damage from vibrational data.
\cite{Mattis2015} applies the multiresponse BBE to a groundwater contamination problem. 
Eight parameters (including source location coordinates, the contaminant source flux, etc.) are inferred from seven data points (the concentrations of the contaminant in seven wells). 
Similarly, \cite{Butler2015} describes the SIP for an application in hydrodynamic models. The authors infer two parameters---the two-dimensional vector of the Manning's $n$ parameter---given maximum water elevations at two of twelve possible observation stations. Note that, although data is available from all twelve stations, only one or two may be used in any given SIP, due to the constraint that $p \geq q$. 
To address this peculiarity of the method, the authors introduce the notion of ``geometrically distinct'' data (or quantities of interest, ``QoI''): since $q \leq p$, then each data point should contain as much information as possible (analogous to linear independence). 
At the same time (and as pointed out by the authors), two different sets of geometrically distinct data may lead to very different solutions for the otherwise same SIP.
\cite{Graham2017} also estimates probable Manning’s $n$ fields using a storm surge model and the BBE SIP framework.
Finally, \cite{Presho2017} uses the BBE framework together with the generalized multiscale finite element method for uncertainty quantification within two-phase flow problems.

More recently, \cite{Uy2019} provides useful examples and commentary on SIPs, with an emphasis on \cite{Breidt2011}; we comment further on this paper in Section~\ref{sec:OtherWork}.

\subsection{The BBE Solution to the SIP is a Change-of-Variables}   \label{sec:BBEisCOV}

In their solutions to the univariate and multivariate inverse sensitivity problem, \cite{Breidt2011} and \cite{Butler2014} provide algorithms to discretize a probability density that is never actually given in closed form.
We begin our reanalysis by deriving what this density must be.
The derivation will clarify and extend arguments of \cite{Uy2019} (Sec.~2) without the explicit use of the disintegration theorem.

Some notation and concepts used within the BBE approach are needed.
In the BBE setup there is a transverse parameterization $\bd t(\bd\theta)$ of dimension $q$, and a $(p-q)$-dimensional contour parameterization $\bd c (\bd\theta)$ of the forward map contours $\bd g$ in $\mathcal Q$.
Furthermore, there exists a one-to-one function $\bd g^t$ between $\bd t(\bd\theta)$ and $\mathcal Q$ induced by $\bd g$, although in practice this function is typically unknown.
The function $\bd g^t$ simply transforms the transverse coordinates uniquely to the observable space $\mathcal Q$ and thus operates as a $q$-dimensional forward map.


It is easiest to derive the density of the BBE solution by first thinking of how to sample from it, and this can be done akin to Algorithm \ref{alg:Intuitive} for the intuitive solutions.
Given a random sample $\bd y^* \sim f_{\bd Y}$, the solution to $\bd y^* - \bd g^t(\bd t^*) = \bd 0$ is a random draw from the density $f_{\bd T}(\bd t)  =  f_{\bd Y}\big( \bd g^t(\bd t) \big) \big| \first{\bd g^t}{\bd t} \big|$.
Next, a uniformly sampled point along the contour indexed by $\bd t^*$ is a realization $\bd c^* \sim f_{\bd C | \bd T}(\bd c \ | \ \bd T = \bd t^*)$ for the special case of a uniform distribution.
Thus the pair $(\bd t^*, \bd c^*)$ is a random draw from the joint distribution $f_{\bd T, \bd C} = f_{\bd T} \cdot f_{\bd C | \bd T}$.
Putting this back in $\bd \theta$-coordinates requires one more change-of-variables and assumes that the map $(\bd t, \bd c) \mapsto \bd\theta$ has nonvanishing Jacobian a.e.
The resulting density is thus given by the following iterated change-of-variables:
\begin{align}
   f^{\text{BBE}}_{\bd\Theta}(\bd\theta)  
   &=   f_{\bd T, \bd C}\Big( \bd c(\bd \theta), \bd t (\bd \theta) \Big)  \Big| \first{(\bd t, \bd c)}{\bd\theta} \Big| \nonumber \\
   &= f_{\bd T} (\bd t (\bd \theta)) f_{\bd C | \bd T}\Big( \bd c(\bd \theta)\ | \ \bd t (\bd \theta) \Big)  \Big| \first{(\bd t, \bd c)}{\bd\theta} \Big| \nonumber \\
   &=
   \left(  f_{\bd Y}\Big( \bd g^t\big( \bd t(\bd\theta) \big) \Big) \Big| \first{\bd g^t}{\bd t} \Big|_{\bd t(\bd\theta)} \ 
   f_{\bd C | \bd T} \Big( \bd c(\bd\theta)  \ | \ \bd t(\bd\theta) \Big)  \right)  \left| \first{(\bd t, \bd c)}{\bd\theta} \right|   \label{eq:BBE}
\end{align}
for $f_{\bd C | \bd T}$ uniform.
We have just sketched a proof of the following result.
\begin{prop} 
Suppose that $\bd g^t(\bd t)$ (induced by the full forward map $\bd g$) is a continuously differentiable, invertible map from the $q$-dimensional image of $\mathcal P$ under $\bd t(\bd\theta)$ to $\mathcal Q$.
Also suppose that together the transverse and contour parameterizations $\bd t(\bd\theta)$ and $\bd c(\bd\theta)$ jointly map to the parameter space $\mathcal P$ with nonvanishing Jacobian according to (A1).
Then the exact BBE solution is equivalent to an iterated change-of-variables given by (\ref{eq:BBE}).
\label{prop:BBEisCoV}
\end{prop}

%

The main difficulty with directly using the CoV form of the BBE solution above is the lack of closed-form transverse and contour parameterizations.
Hence, examples in which the exact density $f^{\text{BBE}}(\bd\theta)$ is available are hard to come by.
The following two examples do however admit closed-form solutions and illustrate the concepts above.
\subsection{Example: Linear Map}\label{ssec:lin-map}

Here we find the exact solution density corresponding to Section~3.1 of \cite{Butler2014}. 
Let $\bd g(\bd\theta) \defi  \bd A \bd\theta$ so that $\bd Y = \bd A \bd\Theta$.
Suppose that $\bd A^\perp$ is a $(p-q)\times p$ representation of the null space of row($\bd A$).

Consider the transverse parameterization that follows the Jacobian of the forward map, namely, $\bd t(\bd\theta) \defi \bd A \bd\theta$.
The connection between the data-space and transverse coordinates is the trivial one: $\bd y = \bd g^t(\bd t) \defi \bd I_q \bd t$.
Thus we have $f_{\bd T}(\bd t) = f_{\bd Y}(\bd t)$.
Contours of $\bd g$ are described by $\bd c(\bd\theta; \bd t) \defi \bd A^\perp \bd\theta$.
The choice of uniformly distributed contours implies that
$f_{\bd C | \bd T}(\bd c | \bd t)  \propto  \mathbbm 1 \big\{ \bd l(\bd t) \leq \bd c \leq \bd u(\bd t) \big\}$, where $\bd l$ and $\bd u$ are given lower and upper bounds (in order to define a valid probability distribution).
If the pre-specified bounds are aligned with $\bd A, \bd A^\perp$, then $\bd l$ and $\bd u$ will not depend on $\bd t$.

Transforming $(\bd T, \bd C)$ to $\bd\Theta$ comes from the invertible augmented system 
\begin{align*}
   \begin{bmatrix}  \bd t \\ \bd c  \end{bmatrix}  
   &=  
   \left( \bd A_{\text{aug}} \defi  \begin{bmatrix}  \bd A \\ \bd A^\perp  \end{bmatrix} \right)  \bd\theta     
\end{align*}
having constant Jacobian $\left|\bd A_{\text{aug}}\right| $, and this implies the result
\begin{align*}
   f^{\text{BBE}}_{\bd\Theta}(\bd\theta)  &\propto  f_{\bd Y}(\bd A \bd\theta) \ \mathbbm 1 \big\{ \bd l(\bd A \bd\theta) \leq \bd A^\perp \bd\theta \leq \bd u(\bd A \bd\theta) \big\} 
   \ .
\end{align*}

The top row of Figure \ref{fig:BBE} shows samples of the BBE solution overlaid upon contours of the forward map $g(\bd\theta) \defi [-1/3, 4/3] \bd\theta$ (left).
The observable density $f_Y(y)$ is a $N\big(1/2, (1/4)^2\big)$ truncated to (0,1); this is plotted over the pushforward histogram of the BBE solution samples (right), confirming the solution.

\begin{figure}[h!t]   \begin{center}
\includegraphics[width=0.8\textwidth]{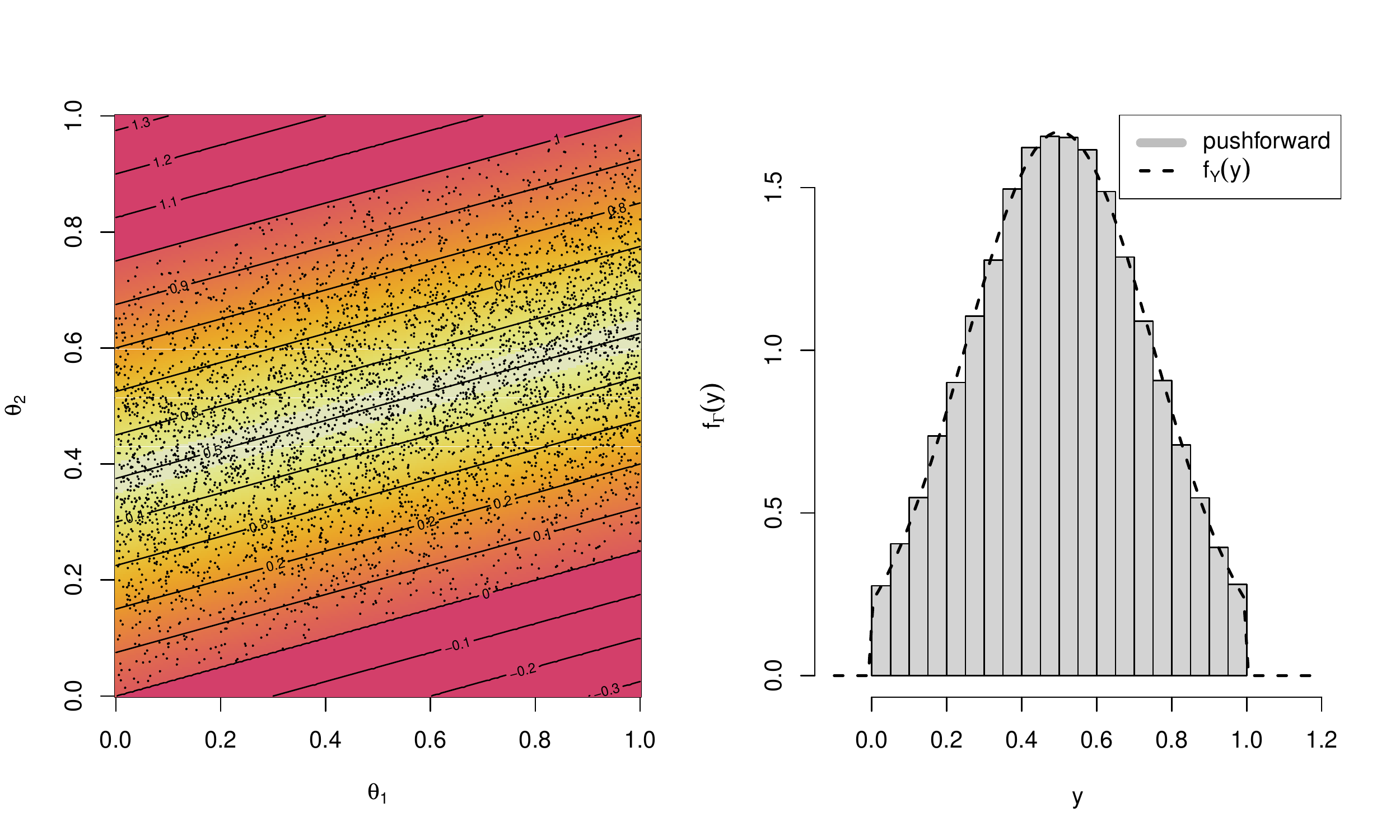} \\
\vspace*{-30pt}
\includegraphics[width=0.8\textwidth]{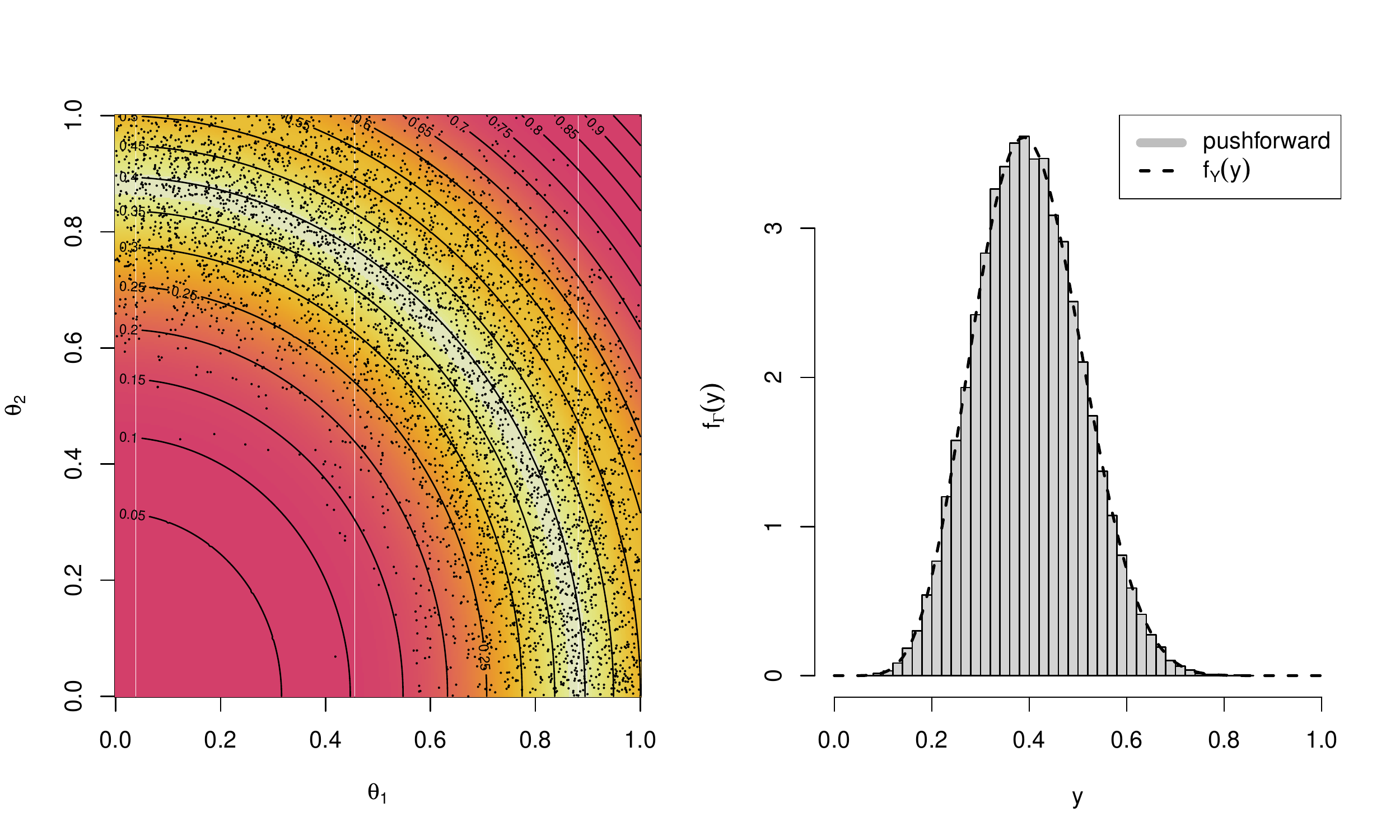}
\caption{
Left: Samples of $\bd \Theta \sim f^{\text{BBE}}_{\bd\Theta}$ and heatmap of $f^{\text{BBE}}_{\bd\Theta}(\bd\theta)$ with contours of the function $g(\bd\theta)$ overlaid; 
Right: The pushforward of the BBE solution samples ($g(\bd\Theta)$, histogram) compared to the given density $f_Y(y)$.
Top row: Linear map in Section~\ref{ssec:lin-map}; 
Bottom row: Nonlinear map in Section~\ref{ssec:nlin-map}. 
}
\label{fig:BBE}
\end{center}
\end{figure}

\subsection{Example: Nonlinear Map}\label{ssec:nlin-map}

Now we work through Example~2 of \cite{Butler2014}. 
Let $Y = g(\bd\Theta)$ with $g(\theta_1, \theta_2)  \defi  \tfrac{1}{2}\big( \theta_1^2 + \theta_2^2 \big)$ with
$0 < \theta_1, \theta_2 \leq 1$.  
(Note we have slightly modified the forward map by introducing the factor of $1/2$, simply to clean up the results.)    
Suppose $f_Y(y)$ is a pdf defined on $(0,1)$, such as a $Beta(a,b)$ distribution.

The symmetry of $g(\bd\theta)$ allows for tidy transverse and contour parameterizations via polar coordinates.
Consider the transverse parameter defined by the distance to the origin
$r(\bd\theta) \defi \sqrt{\theta_1^2 + \theta_2^2}$ (i.e., the radius), and
the contour parameter defined by the polar angle to the positive $\theta_1$-axis,
$\phi(\bd\theta) \defi \text{atan2}(y=\theta_2, x=\theta_1)$.
Observe that the polar angle is increasingly restricted as $r$ goes from 1 to $\sqrt{2}$.

The connection between the data-space and transverse coordinates is $y = g^t(r) \defi \frac{1}{2} r^2$.
Thus we have $f_{R}(r) = f_{Y}(\frac{1}{2} r^2) \cdot r$.
Here, the choice of uniformly distributed contours implies that
\begin{align*}
   f_{\Phi | R}(\phi | r)  &=  
   \left\{  \begin{array}{r l}
   \frac{2}{\pi} \quad \mathbbm 1 \big\{ \ 0 \ < \phi < \ \tfrac{\pi}{2} \big\}  &  \quad 0 < r \leq 1  \\
   \frac{1}{\phi_2 - \phi_1} \ \mathbbm 1 \big\{ \phi_1 < \phi < \phi_2 \big\}  &  \quad 1 < r \leq \sqrt{2} 
   \end{array}  \right.  \\
   \phi_1  &=  \text{atan2}\left( y = \sqrt{r^2 - 1}, x=1 \right)  \\  
   \phi_2  &=  \text{atan2}\left( y = 1, x=\sqrt{r^2 - 1} \right). 
\end{align*}
Transforming $(R,\Phi)$ to $\bd\Theta$ comes with Jacobian $\frac{1}{r}$ evaluated at $r = \sqrt{\theta_1^2 + \theta_2^2}$, and therefore
\begin{align*}
   f^{\text{BBE}}_{\bd\Theta}(\bd\theta)  &=  f_Y\Big( \tfrac{1}{2}(\theta_1^2 + \theta_2^2) \Big)  \  
   f_{\Phi | R}\Big( \text{atan2}(\theta_2, \theta_1) \ | \ \sqrt{\theta_1^2 + \theta_2^2} \Big)  \ .
\end{align*}

The bottom row of Figure \ref{fig:BBE} shows the solution with contours of the nonlinear forward map (left).
The observable density $f_Y(y)$ is a $Beta(8,12)$, and this is plotted over the pushforward histogram of the BBE solution samples (right).

\section{Butler et al. 2018 (BJW) and Related Work}  \label{sec:BJW}

Butler, Jakeman, and Wildey (``BJW'') \citep{Butler2018a} consider the SIP (\ref{eq:InvProb}) as above and call the use of their solution ``consistent Bayesian inference'' or ``pushforward based inference''.
Here the use of the term ``consistent'' does not refer to the statistical limiting sense, but rather to the fact that the solution pushes forward to a given density.
(We prefer neither of these terms since we will show that this BJW solution is neither Bayesian [Sections \ref{sec:BJWupdating}, \ref{sec:BJWisCOV}] nor inference [Section \ref{sec:Inference}]).
The exact solution to the SIP proposed by \cite{Butler2018a} is
\begin{align}
  f_{\bd\Theta}^{\text{BJW}} \big( \bd\theta \big)  &= 
  \ut f_{\bd\Theta} \big( \bd\theta \big)
  \frac{ f_{\bd Y}\big( \bd g(\bd\theta) \big) }{ \ut f_{\bd\Gamma}\big( \bd g(\bd\theta) \big) } \ ,  \label{eq:BJW}
\end{align}
where $\ut f_{\bd\Theta}$ is a given density and $\ut f_{\bd\Gamma}$ is its pushforward through $\bd g$.  
The approximate solution $\widehat f_{\bd\Theta}^{\text{BJW}} \big( \bd\theta \big)$ is what gets used in practice, and this has the same form except the density in the denominator is replaced by an approximate density $\widehat{\ut f}_{\bd\Gamma}\big( \bd g(\bd\theta) \big)$.

The BJW solution to the SIP was at least initially called Bayesian because Bayes' Rule was used in the derivation of the solution, and because the form of the solution explicitly features a given initial density $\ut f_{\bd\Theta}(\bd\theta)$ (potentially playing a role similar to a prior distribution) times a weighting function (akin to a likelihood function).
Indeed, one must specify this initial $p$-dimensional density $\ut f_{\bd\Theta}(\bd\theta)$, similar to the choice of a prior distribution during Bayesian inference. 
Unlike standard Bayesian methods, however, the derivation also explicitly invokes the disintegration theorem of measure theory, and the applied solution relies on kernel density estimation for the denominator term.
The practical reliance upon density estimation restricts the application of BJW to very few observable QoI (small $q$).

Subsequent works explore numerical aspects, applications, and extensions of \cite{Butler2018a}.
\cite{Butler2018b} studies the convergence of kernel density approximate solutions to the analytic BJW-derived densities.
\cite{Walsh2018} and \cite{Butler2020a} propose algorithms for optimal experimental design which maximize the expected information gain between initial and updated densities.
\cite{Butler2020c} applies the SIP framework to a drum manufacturing process. 
The QoI are two (of twenty possible, observable) eigenmodes of the drum vibration, and the parameters are two diffusion parameters.
\cite{Butler2020d} generalizes the BJW solution to ``stochastic'' forward maps; 
we will discuss this generalization further in Section \ref{sec:Butler2020} after an example in which (\ref{eq:BJW}) has closed form.
\cite{Bruder2020} uses multi-fidelity methods and Gaussian process regression models to efficiently solve the SIP.
\cite{Tran2021} focuses on a materials science application, and also augments the BJW approach with a regression model based on Gaussian processes to decrease the computational expense.
Finally, focusing on problems yielding large amounts of time-series data, \cite{Mattis2022} learn the parameter-to-observable (QoI) map from data, thus allowing for specification of the observable probability distribution. Learning this map may rely on clustering the data and creating a partition on the parameter-space. 
Once this is completed through the ``Learning Uncertain Quantities'' framework, then the SIP can be solved as in previous works above.

\subsection{Example: Linear Transformation of a Multivariate Gaussian Vector}  \label{sec:ExampleLinearBJW}

Here we provide an analytic solution to a problem that was partly solved by \cite{Butler2020a}.
Let $\bd g(\bd\theta) \defi \bd A \bd\theta$ so that $\bd Y = \bd A \bd\Theta$.
Suppose that $\bd Y \sim N_q(\bd\mu_y, \bd\Sigma_y)$ and an initial density for $\bd\Theta$ is also multivariate Gaussian: $\ut f_{\bd\Theta}(\bd\theta) \sim N_p(\bd\mu_\theta, \bd\Sigma_\theta)$.
The pushforward of $\ut f_{\bd\Theta}$ under $\bd g$ has distribution
$\bd\Gamma \sim N_p(\bd A\bd\mu_\theta, \bd A \bd\Sigma_\theta \bd A^\top)$.
Using these three given densities, the BJW solution (\ref{eq:BJW}) is
\begin{align*}
   f^{\text{BJW}}_{\bd\Theta}(\bd\theta)  &\propto   
   \exp\left\{ -(\bd\theta - \bd\mu_\theta)^\top \bd\Sigma_\theta^{-1} (\bd\theta - \bd\mu_\theta) /2 \right\}  \\
   & \ \ \cdot \  \frac{ \exp\left\{ -(\bd A\bd\theta - \bd\mu_y)^\top \bd\Sigma_y^{-1} (\bd A\bd\theta - \bd\mu_y) /2 \right\} }{ \exp\left\{ -(\bd A\bd\theta - \bd A\bd\mu_\theta)^\top (\bd A \bd\Sigma_\theta \bd A^\top)^{-1} (\bd A\bd\theta - \bd A\bd\mu_\theta) /2 \right\} }
\end{align*}
After expanding, combining terms, and completing the quadratic form, it is seen that the answer must be
\begin{align}
   f^{\text{BJW}}_{\bd\Theta}(\bd\theta)  &=  N_p(\wtilde{\bd\mu}, \wtilde{\bd\Sigma})  \nonumber  \\
   \wtilde{\bd\mu}  &=  \wtilde{\bd\Sigma} \left(  \bd A^\top \bd\Sigma_y^{-1} \bd\mu_y  
               -  \bd A^\top (\bd A \bd\Sigma_\theta \bd A^\top)^{-1} \bd A \bd\mu_\theta  
               +  \bd\Sigma_\theta^{-1} \bd\mu_\theta \right)   \label{eq:muti} \\
   &= \wtilde{\bd\Sigma}  \bd A^T \bd\Sigma_y^{-1}(\bd \mu_y -\bd A\bd\mu_\theta) + \bd\mu_\theta  \\
   \wtilde{\bd\Sigma}  &=  \left( \bd A^\top \bd\Sigma_y^{-1} \bd A  
                  -  \bd A^\top (\bd A \bd\Sigma_\theta \bd A^\top)^{-1} \bd A  
                  +  \bd\Sigma_\theta^{-1}  \right)^{-1}   \label{eq:sigti}  \\
   &=  \bd\Sigma_\theta - \bd\Sigma_{\theta} \bd A^T\left\{(\bd\Sigma_y^{-1} - (\bd A\bd\Sigma_\theta \bd A^T)^{-1})^{-1} + \bd A \bd \Sigma_\theta \bd A^T \right\} ^{-1} \bd A \bd\Sigma_\theta
\end{align}
Note that when $\bd A$ is a square invertible matrix, $\wtilde{\bd\Sigma} = \bd A^{-1} \bd\Sigma_y \bd A^{-\top}$, and $\wtilde{\bd\mu} = \bd A^{-1} \bd\mu_y$. 
Otherwise, as shown in the appendix (Lemma \ref{lem:BJW}), $\bd A \wtilde{\bd\mu} = \bd\mu_y$ and $\bd A \wtilde{\bd\Sigma} \bd A^\top = \bd\Sigma_y$, confirming that the pushforward indeed has the distribution $\bd A\bd\Theta \sim N_q(\bd\mu_y, \bd\Sigma_y)$.

\subsection{Sequential Updating(?)}  \label{sec:BJWupdating}

With its apparent ability to update an initial distribution through (\ref{eq:BJW}), the BJW solution has been compared to (and presented as an alternative to) Bayesian inference.
We show however that such comparisons are misleading.

Suppose that an analyst wants to compare two $q$-vectors $\bd Y_1, \bd Y_2$ of data to a model $\bd g$ that produces $q$ comparable outputs.
The model takes $p$ inputs through the vector $\bd\theta$, and the goal is to estimate these unknown quantities.
Because $\bd g$ is a well-defined function, it cannot simultaneously map $\bd\theta$ to both $\bd Y_1$ and $\bd Y_2$.
The analyst therefore decides to use the BJW method to update an initial distribution $\ut f_{\bd\Theta}$ before updating it once more so that both pieces of information can be used.
This procedure yields a final answer
\begin{align*}
   \ut{\ut f}_{\bd\Theta} \big( \bd\theta \big)
   \frac{ f_{\bd Y_2}\big( \bd g(\bd\theta) \big) }{ \ut{\ut f}_{\bd \Gamma}\big( \bd g(\bd\theta) \big) }
   &= 
   \left( \ut f_{\bd\Theta} \big( \bd\theta \big)
   \frac{ f_{\bd Y_1}\big( \bd g(\bd\theta) \big) }{ \ut f_{\bd \Gamma}\big( \bd g(\bd\theta) \big) } \right)
   \frac{ f_{\bd Y_2}\big( \bd g(\bd\theta) \big) }{ \ut{\ut f}_{\bd \Gamma}\big( \bd g(\bd\theta) \big) } \\
   &=  
   \ut f_{\bd\Theta} \big( \bd\theta \big)
   \frac{ f_{\bd Y_2}\big( \bd g(\bd\theta) \big) }{ \ut f_{\bd \Gamma}\big( \bd g(\bd\theta) \big) }
\end{align*}
since $\ut{\ut f}_{\bd\Gamma} \equiv f_{\bd Y_1}$ (as $\ut{\ut f}_{\bd\Theta}$ solved the initial SIP).
This is the BJW solution using $\bd Y_2$, and as such, the final answer does not depend on $\bd Y_1$!
If done in reverse order, the final answer would not depend upon $\bd Y_2$.
In a Bayesian analysis, both observations would be used.
(Thankfully, after 2018, the terminology in the literature related to BJW moved away from the name ``consistent Bayes'' and the analogy of Bayesian inference.)

Thus sequential updating is not a defensible way to deal with \emph{replicate} measurements within an SIP.
The next section explores a recent method that allows for replicates.

\subsection{The Extension of BJW to ``Stochastic Maps'' is Actually an Application of BJW}  \label{sec:Butler2020}

In \cite{Butler2020b}, the authors extend the BJW framework to so-called ``stochastic maps''---forward maps that include either embedded or additive parameters meant to capture irreducible aleatoric uncertainty. In their terms, if $\bd g(\bd \theta)$ is a deterministic map, then the corresponding stochastic map $\hat{\bd g}(\bd \theta, \bd \epsilon)$ incoporates either embedded or additive noise. In the additive case, we then have the following system defining the SIP:
\begin{align}
   \bd Y  &=  \hat{\bd g}(\bd\Theta, \bd E)  \quad\quad  
   \hat{\bd g}(\bd\theta, \bd\epsilon) = \bd\theta + \bd\epsilon
\end{align}
%
%
But this is just a special case of BJW. 
That is, there is nothing inherently stochastic in how the new $\bd \epsilon$ term interacts with $\bd g$; if $\bd g(\cdot)$ is deterministic, then so is $\hat{\bd g}(\bd\theta, \bd\epsilon)$. 
It is when the parameters are turned into random variables ($\bd\theta \rightarrow \bd\Theta$ and $\bd\epsilon \rightarrow \bd E$) that the response becomes stochastic, and this has nothing to do with $\hat{\bd g}$.

The addition of these $\bd\epsilon$ ($\bd E$) parameters simply accommodates the case when there are more observables than original parameters of interest $\bd\theta$.
As such, this redefined forward map allows the BJW solution to be used when there are replicate measurements on the same observable.
The BJW solution for the system above gives the ``discrepancy'' tradeoff between the $p+q$ random parameters $\bd\Theta$ and $\bd E$ such that the marginal propagated density $f_{\bd Y}$ is unaffected.



\subsubsection{Gaussian Mean Estimation as ``Stochastic Map'' Inversion}

To see the difference between the ``stochastic map'' BJW solution and a statistical solution, consider the simplest possible scenario where the goal is to estimate a single parameter of interest $\mu$ ($q = 1$) from $n$ independent replicate samples.
The statistical model is $(Y_i | \mu) = \mu + E_i$ for $E_i \stackrel{\text{iid}}{\sim} N(0,\sigma^2_\epsilon)$ ($i = 1, \ldots, n$), or equivalently, 
$\bd Y \sim N_n(\mu \bd 1_n, \sigma^2_\epsilon I_n)$ for $\sigma^2_\epsilon$ given. 
Within the BJW framework, the map is $\hat{\bd g}(\mu, \bd\epsilon) = \mu \bd 1_n + \bd\epsilon$ so that $\bd Y = M \bd 1_n + \bd E$ (where $M$ is the random variable form of $\mu$).

The stochastic BJW solution comes from a system of $n$ equations in $n+1$ unknowns.
Note that if any $E_i$ is fixed at a constant value or given a known distribution, the system becomes square and $M$ is then determined by a single observable.
For example, if $E_1$ is set, then $Y_1$ completely determines the density for $M = Y_1 - E_1$; adding more observables $Y_{n+1}, \ldots$ only provides information on $\epsilon_{n+1}, \ldots$.
In other words, the information on $M$ is contained within a single, arbitrary $Y_1$ and does not increase as more information is gathered.
However, if no $E_i$ is set, then the BJW solution for the $n+1$ parameters instead produces a distribution for $M$ that converges to a point mass.

The forward map is $\bd g(\bd\theta) \defi \bd A \bd\theta$ where $\bd A \defi [\bd 1_n \vdots \bd I_n]$ and $\bd\theta_n \defi (\mu, \epsilon_1, \ldots, \epsilon_n)^\top$.
The observable distribution is $\bd Y \sim  N_n(\bd\mu_y, \bd\Sigma_y)$ with $\bd\Sigma_y = \sigma^2_y \bd I_n$.
When the initial distribution is $\ut f_{\bd\Theta}  \equiv  N_{n+1}(\bd\mu_\theta, \bd\Sigma_\theta)$ with $\bd\mu_\theta^\top  =  (\mu_0, \bd 0_n^\top)$ and $\bd\Sigma_\theta  =  \text{diag}(\sigma_0^2, \sigma^2_\epsilon, \ldots, \sigma^2_\epsilon)$,
the BJW solution is $N_{n+1}(\wtilde{\bd\mu}, \wtilde{\bd\Sigma})$ with mean and covariance given by (\ref{eq:muti}, \ref{eq:sigti}).
We derive the closed-form moments in Section \ref{sec:App2} of the appendix.
The matrix algebra is tedious but reveals some interesting features of the BJW solution, namely that the mean and variance take the form
\begin{align}
   \wtilde{\bd\mu}  
   &=
   \begin{bmatrix}  \overline\mu_y + \mathcal O(n^{-1})  \\  (\bd\mu_y - \overline\mu_y \bd 1_n) +  \mathcal O(n^{-1}) \bd 1_n  \end{bmatrix}  
   \quad\quad
   \wtilde{\bd\Sigma}  =  
   \begin{bmatrix}  \mathcal O(n^{-1})  &  -\mathcal O(n^{-1}) \bd 1_n^\top  \\  -\mathcal O(n^{-1}) \bd 1_n  &  \sigma^2_y \bd I_n - \mathcal O(n^{-1}) \bd J_n  \end{bmatrix}  \ . 
\end{align}
The marginal variance of $M$ decreases like $\mathcal O(n^{-1})$, implying that the distribution of $M$ concentrates on $\overline\mu_y$, a simple average of the means of $\bd Y$.
Marginally, $\bd E$ converges to a distribution with covariance $\sigma^2_y \bd I_n$ that is centered on a residual vector.
Thus all the irreducible uncertainty is contained within $\bd E$, but not within the parameter of interest, contradicting the intended purpose of the BJW approach.

For any $n$, $\wtilde{\bd\Sigma}$ is a dense matrix, meaning that $M$ and $\bd E$ covary and all elements of $\bd E$ covary with one another.
Furthermore, the mean of $M$ is a linear combination of $\overline\mu_y$ and $\mu_0$.
While this is reminiscent of Bayesian inference, it should be noted that the terms in the BJW solution are considerably more complicated and thus harder to interpret.
In addition, a statistical solution involves simple scalar arithmetic while the BJW solution comes from inverting and multiplying potentially large matrices.


\subsection{The BJW Solution to the SIP is a Change-of-Variables}   \label{sec:BJWisCOV}

Due to its generality, the use of measure theory in describing and solving problems is often preferable.
However, if restricting one's attention to absolutely continuous random variables (A2) in practice, then the language of measure theory can obscure concepts which are quite simple in nature.
In the next proposition we show that the BJW density can be derived in a straightforward manner under Assumption (A1).
The second half of its proof relies upon the notion of \emph{auxiliary variables}.
For some choice (but typically, infinitely many choices) of $p-q$ auxiliary variables $\bd Y^c$ via maps $\bd g^c(\cdot) \defi \left( g_{q+1}(\cdot),\ldots,g_p(\cdot) \right)$, the following augmented system is locally invertible:
\small     
\begin{align}
  \bd\Theta = 
  \left[  \begin{array}{l}
  \Theta_1  \\
  \vdots    \\ 
  \Theta_q  \\[1ex]
  \Theta_{q+1}  \\[1ex]
  \vdots    \\
  \Theta_p
  \end{array}  \right] 
  \quad
  \begin{array}{c}
  \bd g_{\text{aug}}  \\
  \longmapsto  
  \end{array}
  \quad
  \left[  \begin{array}{lcl}
  Y_1      &  \hspace*{-6pt} =     \hspace*{-6pt}  &  g_1(\bd\Theta)  \\
  \vdots   & &  \\ 
  Y_q      &  \hspace*{-6pt} =     \hspace*{-6pt}  &  g_q(\bd\Theta)  \\[1ex]
  \hline
  Y_{q+1}  &  \hspace*{-6pt} \defi \hspace*{-6pt}  &  g_{q+1}(\bd\Theta)  \\
  \vdots   & &  \\
  Y_p      &  \hspace*{-6pt} \defi \hspace*{-6pt}  &  g_p(\bd\Theta) \hspace*{11pt} 
  \end{array}  \right] 
  = \begin{bmatrix} \bd Y \\ \bd Y^c \end{bmatrix}
  \defi  \bd Y_{\text{aug}}  \ .  \label{eq:AuxVars}
\end{align}
\normalsize
For the time being we shall take the existence of $\bd g_{\text{aug}}$ 
defining the locally invertible system as given, but more details will be given in Section \ref{sec:ALLisCov}.

\begin{prop} 
Under assumption (A1) the exact BJW solution can be derived from a CoV solution for $p \geq q$.
\label{prop:BJWisCoV}
\end{prop}
\begin{proof}
Fix some initial $p$-dimensional density $\ut f_{\bd\Theta}(\bd\theta)$.

First consider the case that $p = q$.
%
Using the same reasoning leading to (\ref{eq:PushforwardPullback}), the initial density can be written in terms of its pushforward on $\mathcal P \setminus \mathcal P_0$, regardless of the global invertibility of $\bd g$.
Equivalently, the Jacobian determinant is
\begin{align}
   \Big| \first{\bd g}{\bd\theta} \Big|  &=  
   \frac{ \ut f_{\bd\Theta}(\bd\theta) }{ \ut f_{\bd\Gamma}(\bd g(\bd\theta)) }  
\end{align}
wherever $\ut f_{\bd\Theta} > 0$ (and hence $\ut f_{\bd\Gamma} > 0$) on $\mathcal P \setminus \mathcal P_0$.
Now we have
\begin{align*}
  f_{\bd\Theta}^{\text{BJW}} \big( \bd\theta \big)  &= 
  f_{\bd Y}\big( \bd g(\bd\theta) \big)
  \frac{ \ut f_{\bd\Theta} \big( \bd\theta \big) }{ \ut f_{\bd \Gamma}\big( \bd g(\bd\theta) \big) }
   =  f_{\bd Y}\big( \bd g(\bd\theta) \big) \ \Big| \first{\bd g}{\bd\theta} \Big|  
\end{align*}
which is the density corresponding to a CoV, as desired.
The explicit mixture of CoVs from $\mathcal Q$ to $\mathcal P$ can be constructed exactly as in the proof of Proposition \ref{prop:CoVpq}.
Note that in this $p = q$ case, $\ut f_{\bd\Theta}$ did not affect the solution to the SIP since $\big| \first{\bd g}{\bd\theta} \big|$ is invariant to this choice.

The case of $p > q$ follows exactly the same reasoning as above but relies on auxiliary variables, as per (\ref{eq:AuxVars}).
Additionally, it is slightly more transparent to start with $f^{\text{CoV}}$, and show that for a certain choice, it becomes the exact BJW solution.
Replacing ``$\bd g$'', ``$\bd Y$'', and ``$\bd\Gamma$'' above with ``$\bd g_{\text{aug}}$'', ``$\bd Y_{\text{aug}}$'', and ``$\bd\Gamma_{\text{aug}}$'',
we have by previous reasoning that 
\begin{align*}
   \Big| \first{\bd g_{\text{aug}}}{\bd\theta} \Big|  &=  
   \frac{ \ut f_{\bd\Theta}(\bd\theta) }{ \ut f_{\bd\Gamma_{\text{aug}}}(\bd g_{\text{aug}}(\bd\theta)) }  \\
   f^{\text{CoV}}_{\bd\Theta}(\bd\theta)  &=  
   f_{\bd Y_{\text{aug}}}\big( \bd g_{\text{aug}}(\bd\theta) \big) \ \Big| \first{\bd g_{\text{aug}}}{\bd\theta} \Big|  \\
   %
   %
   &=  \ut f_{\bd\Theta}(\bd\theta) \  
   \frac{ f_{\bd Y_{\text{aug}}}\big( \bd g_{\text{aug}}(\bd\theta) \big) }{ \ut f_{\bd \Gamma_{\text{aug}}}\big( \bd g_{\text{aug}}(\bd\theta) \big) }  \\
   &=  \ut f_{\bd\Theta}(\bd\theta) 
   \frac{ f_{\bd Y}\big( \bd g(\bd\theta) \big) }
        { \ut f_{\bd\Gamma}\big( \bd g(\bd\theta) \big) } \ 
   \frac{ f_{\bd Y^c | \bd Y} \big( \bd g^c(\bd\theta) \ | \ \bd y = \bd g(\bd\theta) \big) }
        { \ut f_{\bd\Gamma^c | \bd\Gamma} \big( \bd g^c(\bd\theta) \ | \ \bd\gamma = \bd g(\bd\theta) \big) }  \ ,
\end{align*}
where $\bd\Gamma_{\text{aug}} = \{\bd\Gamma, \bd\Gamma^c\}$ and $\ut f_{\bd \Gamma_{\text{aug}}}$ is the pushforward of $\ut f$ through $\bd g_{\text{aug}}$.
Setting $f_{\bd Y^c|\bd Y}  \equiv  \ut f_{\bd\Gamma^c|\bd\Gamma}$ leads to cancellation that establishes that 
$f^{\text{BJW}}_{\bd\Theta}(\bd\theta)$ is indeed derivable from a CoV solution.
\end{proof}

In the proof of Proposition above, Bayes' Theorem was not invoked and the disintegration theorem did not need to be called explicitly.
All that was required was that the joint density of unobserved and observed variables $\bd Y_{\text{aug}} = \{ \bd Y, \bd Y^c \}$ could be factored as a conditional times a marginal density.
Of course this fact is directly related to the disintegration theorem, but will nonetheless be easier for most audiences to digest.

The statement at the end of the proof, $f_{\bd Y^c|\bd Y} \stackrel{\text{set}}{=} \ut f_{\bd\Gamma^c|\bd\Gamma}$, may seem mysterious, but actually provides the density to assign to the set of auxiliary variables within a CoV exercise to match the BJW solution.
First choose a complementary map $\bd g^c(\cdot)$ that yields a full-rank augmented system (a.e.).  
Under the augmented map $\bd g_{\text{aug}}$,
the initial density $\ut f_{\bd\Theta}$ pushes forward to the $p$-dimensional joint density $\ut f_{\bd\Gamma, \bd\Gamma^c} = \ut f_{\bd\Gamma^c | \bd\Gamma} \cdot \ut f_{\bd\Gamma}$.
The BJW solution to the SIP is thus a density obtained through a linear combination of CoVs from $\mathcal Q$ to $\mathcal P$
with $\bd Y_{\text{aug}}  \sim  f_{\bd Y}(\bd y) \cdot  \ut f_{\bd\Gamma^c | \bd\Gamma}(\bd y^c | \bd y)$.

In theory, one could sample from $f^{\text{BJW}}_{\bd\Theta}(\bd\theta)$ using Monte Carlo reasoning akin to Algorithm \ref{alg:Intuitive}. 
One would first generate $\bd y^* \sim f_{\bd Y}$, then use this value to obtain a conditional draw $\bd y^{c*} \sim  \ut f_{\bd\Gamma^c | \bd\Gamma = \bd y^*}$ to form $\bd y_{\text{aug}}^* \defi (\bd y^*, \bd y^{c*})$.
Finally, the solution to 
$\bd y_{\text{aug}}^* - \bd g_{\text{aug}} \big( \bd\theta^* \big)  = \bd 0$
would give a realization $\bd\theta^*$ of the BJW solution.
However, this is hypothetical since in practice one does not have the ability to simulate from the unknown conditional density $\ut f_{\bd\Gamma^c | \bd\Gamma}$.

\section{Other Related Work}  \label{sec:OtherWork}

In order to complete our literature review of SIPs, we need to discuss two more recent works: \cite{Swigon2019} and \cite{Uy2019}.

\cite{Swigon2019} demonstrates the critical role of the Jacobian determinant in two types of parameter estimation problems. 
In the first, the data is given by a nonlinear transformation of the parameters, i.e., $\bd Y = \bd g(\bd \Theta)$, and the authors call this situation (i.e., our SIP) a ``random parameter model''.
Here the answer is simply given by the CoV, similar to our conclusion for $p=q$, and as such depends on the Jacobian determinant.
In the second type of parameter estimation problem, the data is corrupted by error, i.e., $\bd Y = \bd g(\bd \theta) + \bd E$, for some fixed $\bd \theta$, which the authors called a ``random measurement error model'' (and we call a statistical IP). 
In this context the Jacobian determinant appears in the construction of a default prior for $\bd\Theta$.
The authors take both these modeling approaches as worthy of consideration, and while offering some discussion akin to our Section \ref{sec:Inference}, do not tie their results into the work stemming from BBE and BJW.

\cite{Uy2019} restricts attention to the case that $p > q$ and approach SIPs from a different angle: given that the problem is underdetermined, what additional knowledge about $\bd \Theta$ is required to pose and solve SIPs? 
(By contrast, traditional IPs can use ``extra'' information to regularize solutions in overdetermined systems, but here represents necessary missing information for the underdetermined system.)
In short, one can either specify moments of $\bd \Theta$ or its distributional family, similar to methods based on the principle of maximum entropy. 
Interestingly, if additional information is provided on the family of distributions to which $\bd \Theta$ belongs, then the new inverse problem can be solved using standard Bayesian methods (see Section~3.2.2).
For each option, the authors explore two possibilities about the data: (1) pdf information about $\bd Y = \bd g (\bd \Theta)$ is known, or (2) samples from $\bd Y$ are given. 
While mainly focusing on the BBE solutions, \cite{Uy2019} stresses the fundamental underdetermined nature of SIPs (previously noted in \cite{Breidt2011}).
We will take this a little further in Section \ref{sec:Inference} (IV).

\cite{Uy2019} shows that the uniform contour ansatz can lead to arbitrarily bad approximations of the true density, and in some cases cannot recover the true density at all. 
Moreover, the authors specify this additional information in order to find the true distribution of $\bd \Theta$, so that it can be queried in other contexts, and not just to match $\bd g(\bd \Theta)$. 
This goal causes a philosophical breach from \cite{Breidt2011} and \cite{Butler2018a}, in which different data points may lead to different solutions, even within the same problem.





\section{All Solutions to SIPs are Changes-of-Variables}   \label{sec:ALLisCov}

Our main proposition follows from the CoV theorem and contains elements of Proposition \ref{prop:BJWisCoV}.
We give the proof here because it is both constructive and instructive.
\begin{prop} 
Let $f_{\bd\Theta}(\bd\theta)$ be any exact solution to the SIP of (\ref{eq:InvProb}), and suppose the forward map $\bd g$ is as described in Assumption (A1).
Then $f$ is derivable from a CoV (a.e.) for at least one choice of auxiliary variables.
\label{prop:CoV}
\end{prop}
\begin{proof}
The $p = q$ case is covered in Proposition \ref{prop:CoVpq}, so we now assume $p > q$.  
After constructing auxiliary variables, the proof follows in much the same fashion.

Under (A1) we have that the component functions of $\bd g$ (i.e., $g_1,\ldots,g_q$) produce a Jacobian with left $q \times q$ block $\big| \first{\bd g}{\bd\theta_{1:q}} \big|$ that is invertible on $\mathcal P_0$.
Now observe that there exist maps $g_{q+1}, \ldots, g_p$ from $\left( \mathcal P \cap \mathbb R^{p-q} \right)$ to $\mathbb R^{p-q}$ such that the corresponding lower-right $(p-q)\times(p-q)$ sub-Jacobian is invertible;
Again, let $\bd g_{\text{aug}}$ denote the full $p$-dimensional function. 
Any choice of these $(p-q)$ functions results in the definition of auxiliary variables $\bd Y^c \defi (Y_{q+1}, \ldots, Y_p)^\top$. 
The simplest augmented system is one that uses the identity maps to define the components of $\bd Y^c$
\small
\begin{align}
  \left[  \begin{array}{l}
  \Theta_1  \\
  \vdots    \\ 
  \Theta_q  \\[1ex]
  \Theta_{q+1}  \\[1ex]
  \vdots    \\
  \Theta_p
  \end{array}  \right] 
  \begin{array}{c}
  \bd g_{\text{aug}}  \\
  \longmapsto  
  \end{array}
  \left[  \begin{array}{lcl}
  Y_1      &  \hspace*{-6pt} =     \hspace*{-6pt}  &  g_1(\bd\Theta)  \\
  \vdots   & &  \\ 
  Y_q      &  \hspace*{-6pt} =     \hspace*{-6pt}  &  g_q(\bd\Theta)  \\[1ex]
  \hline
  Y_{q+1}  &  \hspace*{-6pt} \defi \hspace*{-6pt}  &  g_{q+1}(\bd\Theta) \defi \Theta_{q+1}  \\
  \vdots   & &  \\
  Y_p      &  \hspace*{-6pt} \defi \hspace*{-6pt}  &  g_p(\bd\Theta) \hspace*{11pt} \defi \Theta_p
  \end{array}  \right] 
\end{align}
\normalsize
which produces Jacobian determinant
\small
\begin{align}
  \left| \first{\bd g_{\text{aug}}}{\bd\theta} \right|
    \quad  &=  \quad  \text{det}
  \begin{bmatrix}
    \first{g_1}{\theta_1}      &  \cdots  &  \first{g_1}{\theta_q}  &  
    \first{g_1}{\theta_{q+1}}  &  \cdots  &  \first{g_1}{\theta_p}  \\
    \vdots  &  \vdots  &  \vdots  &  \vdots  &  \vdots  &  \vdots  \\
    \first{g_q}{\theta_1}      &  \cdots  &  \first{g_q}{\theta_q}  &  
    \first{g_q}{\theta_{q+1}}  &  \cdots  &  \first{g_q}{\theta_p}  \\[1ex]
    \hline  
    0  &  \cdots  &  0  &  1  &  \bd 0^\top  &  0  \\
    \vdots  &  \vdots  &  \vdots  &  \bd 0  &  \ddots  &  \bd 0  \\    
    0  &  \cdots  &  0  &  0  &  \bd 0^\top  &  1  \\
  \end{bmatrix}  
    \quad  =  \quad   \left| \first{\bd g}{\bd\theta_{1:q}} \right|  \ .
\end{align}
\normalsize
This augmented system is invertible due to the fact that the upper-left $(q \times q)$ block $\first{\bd g}{\bd\theta_{1:q}}$ is itself invertible.
The pushforward of $f_{\bd\Theta}(\bd\theta)$ is thus possible by the CoV theorem.
Furthermore, because $f_{\bd\Theta}(\bd\theta)$ is assumed to solve the SIP, its pushforward is a $p$-dimensional distribution $f_{\bd Y_{\text{aug}}}$ whose marginal is $f_{\bd Y}(\bd y)$ for the first $q$ entries.
In fact, 
$f_{\bd Y_{\text{aug}}} = f_{\bd Y}\cdot f_{\bd Y^c | \bd Y}$ with 
\begin{align*}
   f_{\bd Y^c | \bd Y}(\bd y^c | \bd y)  \defi  f_{\bd\Theta_{(q+1):p} | \bd\Theta_{1:q}}(\bd y^c | \bd y)
\end{align*}
under the identity maps.

Going the other direction, the reasoning in the proof of Proposition \ref{prop:CoVpq} can again be used. 
Specifically, Assumption (A1) guarantees local inverses $\bd h_1, \ldots, \bd h_d$ of $\bd g$ defined on the appropriate sets, and we have that the posited solution can be written as a particular $d$-combination of CoV solutions from the augmented range to $\mathcal P \setminus \mathcal P_0$ except for on a set of measure zero.
%
%
The presumed solution was then actually a CoV using a particular choice of auxiliary variables $\bd Y^c$.

Thus we have shown that no matter how one has obtained an exact solution to the SIP, it could have actually been derived from CoV solutions for any choice of auxiliary variables to make an augmented forward map $\bd g_{\text{aug}}$ locally invertible.
\end{proof}

Some readers will recognize the augmented system within the proof as that used within standard constructive proof of the Implicit Function Theorem (given the Inverse Function Theorem); see e.g. \cite{Lee2013} pg. 661.
This auxiliary variable strategy is exactly what is taught in a first course on mathematical statistics.
For example, to find the distribution of the product $Y$ of two random variables $\Theta_1$ and $\Theta_2$, one may 
augment the one-dimensional system $Y \defi \Theta_1 \Theta_2$ with the auxiliary variable $Y^c \defi \Theta_2$ (for example), 
obtain $f_{Y,Y^c}$ from the two-dimensional CoV, and integrate the joint density to get $f_Y$.

The proposition above shows that the answer will be completely determined (a.e.) from the choice of $p-q$ dimensional density $f_{\bd Y^c | \bd Y}$.
Intuitive solutions (\ref{eq:Intuitive2}) can be viewed as the simplest cases that use $(p-q)$ identity maps (as in the proof above), together with $f_{\bd Y^c | \bd Y} \equiv f_{\bd\Theta_{(q+1):p}}$ chosen by the analyst.

In order to sample from a CoV density $f^{\text{CoV}}_{\bd\Theta}(\bd\theta)$, it may appear that the Jacobian determinant $\big| \first{\bd g_{\text{aug}}}{\bd\theta} \big|$ (or at the very least $\big| \first{\bd g}{\bd\theta_{1:q}} \big|$) is readily available \citep{Swigon2019}.
However, this is not true.
Instead of using the explicit density 
to obtain random draws, one could instead use a simple Monte Carlo routine similar to Algorithm \ref{alg:Intuitive}. 
One would first generate $\bd y^* \sim f_{\bd Y}$, then use this value to obtain a conditional draw $\bd y^{c*} \sim f_{\bd Y^c | \bd Y = \bd y^*}$.
The vector $\bd y_{\text{aug}}^* \defi (\bd y^*, \bd y^{c*})$ is a random draw from $f_{\bd Y_{\text{aug}}}$.
Finally, the solution to 
$\bd y_{\text{aug}}^* - \bd g_{\text{aug}} \big( \bd\theta^* \big)  = \bd 0$
would give a realization $\bd\theta^*$ of the CoV solution.

\subsection{Example: Linear Transformation of a Multivariate Gaussian Vector}   \label{sec:LinearMVN}

Let us return to the example of Section \ref{sec:ExampleLinearBJW}.
Again, let $\bd g(\bd\theta) \defi \bd A \bd\theta$ so that $\bd Y = \bd A \bd\Theta$; the dimension of $\bd A$ is $q \times p$, and this matrix is assumed to have full row-rank.
Suppose that the observation distribution is $\bd Y \sim  N_q(\bd\mu_y, \bd\Sigma_y)$.

When $p=q$, $\bd A$ is invertible so the SIP is solved by $\bd\Theta = \bd A^{-1} \bd Y$, and this has density
\begin{align}
f^{\text{CoV}}_{\bd\Theta}(\bd\theta)  =  
N_p \Big( \bd A^{-1}\bd\mu \ , \ \bd A^{-1} \bd\Sigma \bd A^{-\top}\Big)  \label{eq:LinearMVN}  \ .
\end{align}
When $p>q$, $\bd A$ can be augmented into an invertible matrix in an infinite number of ways.
Suppose that $\bd A^\perp$ is a $(p-q)\times p$ representation of the null space of row($\bd A$).
The augmented matrix $\bd A_{\text{aug}} \defi  \begin{bmatrix}  \bd A \\ \bd A^\perp  \end{bmatrix}$ is then invertible.
Furthermore $\bd A^\perp$ defines $p-q$ auxiliary variables $\bd Y^c$.
If the joint distribution of $(\bd Y, \bd Y^c)$ is $N_p(\bd\mu_+, \bd\Sigma_+)$ (making $\bd\mu_y$ the first $q$ entries of $\bd\mu_+$ and $\bd\Sigma_y$ the upper left $q \times q$ block of $\bd\Sigma_+$), then the CoV solution to the SIP is
\begin{align*}
  f^{\text{CoV}}_{\bd\Theta}(\bd\theta)  &=  N_p \Big( \bd A_{\text{aug}}^{-1}\bd\mu_+ \ ; \ \bd A_{\text{aug}}^{-1} \bd\Sigma_+ \bd A_{\text{aug}}^{-1}\Big)  \ .
\end{align*}

The density above begs the question: how does one choose a density for \emph{unobservable} quantities in order to augment and close the system?
We return to this fundamental question in the next section.
The issue of underdeterminedness was observed back in \cite{Breidt2011} (pg. 1839)
for the bivariate Gaussian case ($p=2, q=1$) by considering moment conditions (instead of auxiliary variables), but this did not curtail the investigation or application of SIPs.

\section{Changes-of-Variables and Inference}   \label{sec:Inference}

In the previous section we showed that any solution to an SIP (\ref{eq:InvProb}) can be viewed as a CoV.
The solution to the IP given in (\ref{eq:InvProbClassicNoise}) is typically a matter of statistical \emph{inference}.
What is the connection between these two solutions?
In this section we explore this question and more broadly examine the interface of CoVs and inference. 
While multiple works provide some discussion to distinguish SIP solutions from more traditional statistical inference---Section~2 in~\cite{Breidt2011}; Section~7 in~\cite{Butler2018a}; Section~4.3 in~\cite{Butler2014}; Remark~2.3 in~\cite{Butler2020b}; and throughout~\cite{Uy2019}, especially Sections~1 and~4.2---there is still much to address, and we do so here.

In any inferential context, a change of the likelihood function $\propto f_{\bd Y}(\bd y ; \bd\theta)$ with respect to $\bd y$ is obviously a CoV.
In the Bayesian or Fiducial paradigms, a change of the target distribution $f(\bd\theta | \bd y)$ with respect to $\bd\theta$ is a reparameterization that is also quite conspicuously a CoV. 
A less obvious connection between inference and CoVs is the very concept of Bayesian inference itself.
If $f_{\bd\Theta}^{\text{pri}}(\bd\theta)$ is a proper prior distribution, and $f_{\bd\Theta}^{\text{post}}(\bd\theta) \propto f(\bd y | \bd\theta) \cdot  f_{\bd\Theta}^{\text{pri}}(\bd\theta)$ denotes the posterior distribution, there is some ``forward map'' $\bd g_T$ that pushes prior forward to posterior.
In fact this is the motivation behind \emph{transport maps} where, in an inference setting, the goal is to learn exactly that function $\bd g_T$. This CoV is however entirely distinct from the SIP solutions discussed here: with transport maps, the function $\bd g_T$ maps \emph{from parameter to parameter} space, not from parameter to observable space
\citep{Marzouk2016, Baptista2021}. 

The solutions to (\ref{eq:InvProbClassicNoise}) and (\ref{eq:InvProb}) can sometimes agree, but as we will argue, this does not mean that statistical IPs and SIPs are analogous notions.
We will first give examples where Bayesian posterior and Generalized Fiducial (GF) distributions are identical to a CoV before giving five (interrelated) ways in which SIPs are fundamentally different from inference.
We close by providing two illustrative examples.

\textbf{Bayesian and GF solutions can appear to be CoVs when $p = q$.}
Consider again the linear map to a Gaussian observable (Section \ref{sec:LinearMVN}) having known covariance $\bd\Sigma_y$.
When $p=q$ and an improper prior is chosen $f_{\bd\Theta}^{\text{pri}}(\bd\theta) \propto 1$ for $\bd\Theta$, the posterior is
$f_{\bd\Theta}^{\text{post}}(\bd\theta)  \propto  f_{\bd Y}(\bd y ; \bd A \bd\theta) \cdot 1 $
%
%
which agrees with the CoV solution (\ref{eq:LinearMVN}) after a simple rearrangement.
A similar result is also obtained in the appendix of \cite{Swigon2019}.
The density above is also the uncertainty distribution in the GF paradigm \citep[Ex. 3]{Hannig2016}.
In fact, when $p=q$ and $\bd g$ is an invertible nonlinear map, the GF solution is $f_{\bd\Theta}^{\text{GF}}(\bd\theta)  \propto  f_{\bd Y}(\bd y ; \bd g(\bd\theta)) \big| \first{\bd g}{\bd\theta} \big|$
%
%
\citep[Eqns 3,4]{Hannig2016} 
which is the CoV solution (\ref{eq:CoVsolution}).
The agreement in these special cases is superficial, as a closer look in the first point below will show a sharp contrast in the interpretation of the terms involved.

\textbf{(I) In SIPs, observables are \emph{populations} with given parameters.}
It is in the application of the CoV methodology that one can clearly see the difference between SIPs and any inferential framework based upon a likelihood function.
In an inferential context, after the data is observed, $\bd Y$ becomes $\bd y_{\text{obs}}$ and the data pdf $f_{\bd Y}(\bd y ; \bd\theta)$ will yield the likelihood $L(\bd\theta ; \bd y_{\text{obs}})$.
The likelihood is a function of $\bd\theta$ and indexed by the data $\bd y_{\text{obs}}$. In the Gaussian examples above, the \emph{unknown} mean of $\bd Y$ features the parameters: $\bd A \bd\theta$ or $\bd g(\bd\theta)$.
On the other hand, the CoV solution is based upon the density for the observables $f_{\bd Y}(\bd y)$ which is actually $f_{\bd Y}(\bd y ; \bd y_{\text{obs}})$ where the observed ``data'' actually operate as \emph{known} population parameters!
This density has both $\bd y$ and $\bd y_{\text{obs}}$ as arguments, and the solution will have the $\bd y$ argument replaced with the forward map: $f_{\bd Y}(\bd y = \bd g(\bd\theta) ; \ \bd y_{\text{obs}})$.
The agreement of the Bayesian/GF and CoV solutions for the $p = q$ cases above is due to the symmetry of the Gaussian density with respect to its $\bd y$ and mean arguments.
Moreover, SIPs require all such observable population parameters to be given: as seen above, SIPs involving Gaussian observables require the covariance $\bd\Sigma_y$ to be given.
This is quite different from statistical IPs where quite often covariances are parameters to estimate.

SIP solutions are not inference because they essentially treat everything on a ``population'' level to begin with: the observables \emph{are not a realization} from an unknown population density, they \emph{are the population} with a known density.  
In the case of inference, the data model $f_{\bd Y}(\bd y; \bd g(\bd\theta))$ features the forward map and population parameters $\bd\theta$ to infer.
The likelihood function gives desirable large-sample properties such as the consistency and asymptotic normality of the maximum likelihood estimator \citep{Reid2010}.
In the case of SIP solutions, the density for the observables $f_{\bd Y}(\bd y ; \bd y_{\text{obs}})$ is a population-level description with known parameters and depends upon neither $\bd g$ nor $\bd\theta$. 
As such, $f_{\bd Y}(\bd y)$ does not feature any $\bd g(\bd \theta)$ in it, so can neither be interpreted as a conditional density nor a likelihood given the map $\bd g(\bd \theta)$.  
Hence there is no likelihood function or the corresponding theory.

In fact, the phrases ``increasing sample size'' and ``collecting more data'' are not really even coherent phrases in the world of SIPs; adding observables is simply changing the population to invert.
Despite the fact that the BJW solution was first presented as ``consistent Bayesian inference'', it too suffers from from this changing-populations issue because none of the $f_{\bd Y}$'s for the new data explicitly contain any $\bd g(\bd \theta)$ terms. 
In contrast, within any statistical inference procedure, each new piece of data will contribute knowledge of $\bd \theta$ via $\bd g(\bd \theta)$ through a joint likelihood function.

\textbf{(II) SIPs require $p \geq q$.}
Perhaps the most obvious sign that an SIP is different from a statistical IP is the hard stipulation that $p \geq q$, i.e., more parameters than observables or data.
On the other hand, in any likelihood-based inference setting (Bayesian, GF, frequentist, etc.), an analyst takes great care to avoid unregularized or saturated models and the inherent risk of ``over-fitting''.
This is tied to the problem of \emph{prediction}.

Along with obtaining an uncertainty distribution for $\bd\Theta$, another common goal is prediction: inferring a new value given those values already observed.
In the Bayesian/GF frameworks, the (posterior) predictive density is
$f_{\bd Y^* | \bd Y} = \int f_{\bd Y^* | \bd\Theta, \bd Y} \ f_{\bd\Theta | \bd Y} d\bd\theta$.
This density is trustworthy when the models $f_{\bd Y | \bd\Theta}$ and $f_{\bd Y^* | \bd\Theta, \bd Y}$ are adequate.
Adequacy can be assessed by analyzing residuals and out-of-sample predictions.
When $p \geq q$, a non-Bayesian introduces regularization or smoothness penalties (such as in LASSO) to better constrain the problem, and a Bayesian codifies complexity penalization into the prior specification.
In SIPs there is no way to accommodate such regularization: all solutions embody over-fitting.
Prediction, however, is still possible within the CoV framework, through some new forward map $\bd g^*$.
After obtaining $\bd\Theta \sim f_{\bd\Theta}^{\text{CoV}}$, a predictive density is simply $\bd g^*(\bd\Theta)$.
However, as $\bd Y \rightarrow \bd\Theta \rightarrow \bd g^*(\bd\Theta)$ is a sequence of population-level reparameterizations (I), there are no falsifiable models that can be checked, and hence no way to defensibly justify prediction.

The SIP framework has no extension into the realm of $p < q$; instead, one has to force a problem into becoming an SIP and out of the realm of inference.
Approaches to ensure that $p \geq q$ include the following.
First, the analyst can pick a new subset of $p$ or fewer observables, even though each subset will lead to a different answer. 
This is the approach followed in~\cite{Butler2015}, in which 2 of 12 possible water levels are selected, and in~\cite{Butler2020c}, in which 2 of 20 possible eigenmodes of a drum are selected. 
A related approach is to combine the $q$ original observables into $p$ or fewer, such as by taking averages and pooling variances.
Per the LUQ framework developed by~\cite{Mattis2022}, large quantities of time series data are transformed into samples of a small number of QoIs (via PCA-based feature extraction).
Alternatively, the analyst can add parameters; the original forward map is effectively augmented to make $p \geq q$.
One way to do this is to add a new parameter for each new observable, as suggested in the stochastic map framework~\cite{Butler2020b}. 
In this case, the number of unknowns is inflated from $p \mapsto p+q > q$.

\textbf{(III) SIPs accommodate replicates in an awkward fashion.}
This is very related to (II) above but can be an issue regardless of $p$'s relation to $q$.
Suppose for example that a computer model has $p$ parameters ($\bd\theta$) as inputs and $q$ outputs. 
If there were $n$ vectors of measurements that can be related to the computer model output, one could not immediately treat the estimation of $\bd\theta$ as an SIP.
Either the data vectors would have to be combined into a single $q$-vector (to be treated like a population, I), or following \cite{Butler2020b}, $n$ additional $q$-vectors of unknowns $\bd\epsilon_i$ would be added ($p \mapsto p + nq$).

This awkward, forced adjustment contrasts the world of inference where replicate measurements are a virtual cornerstone of statistical design and analysis of experiments.

\begin{figure}[hbt!]   \begin{center}
\includegraphics[width=0.8\textwidth]{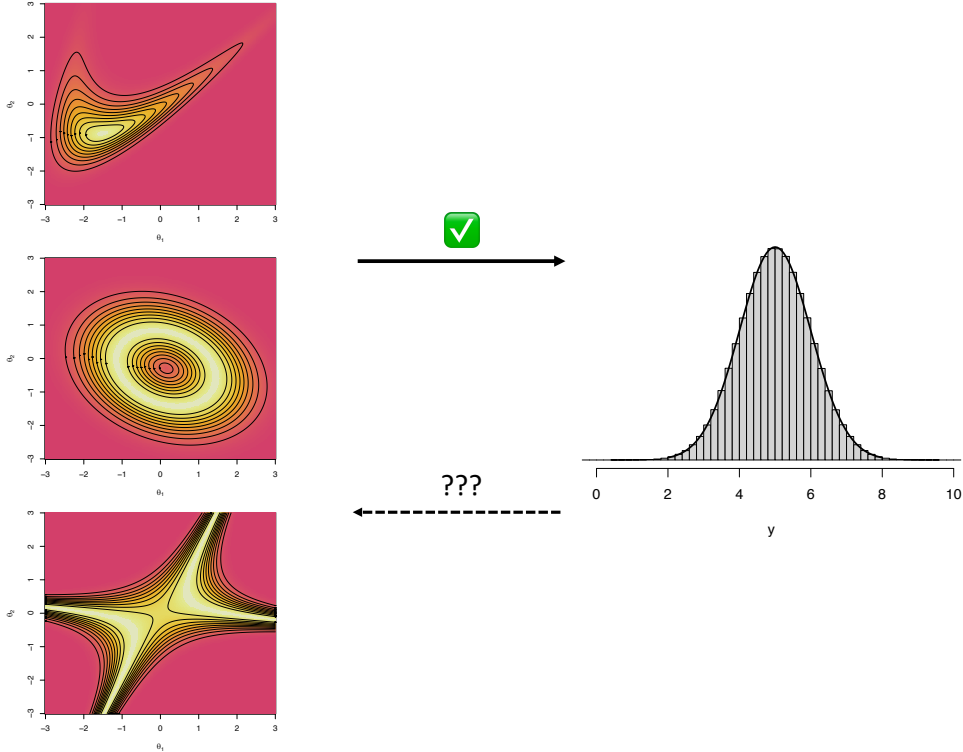} 
\caption{
A demonstration of the fundamental underdetermined nature of the SIP for $p=2$ and $q=1$. 
The only way to recover the true $f_{\bd\Theta}$ distributions (left) from given only $f_{Y}$ (right) is to use unknown (and unknowable) auxiliary variables whose joint distribution could be arbitrarily complex.
}
\label{fig:ForwardInverse}
\end{center}   \end{figure}

\textbf{(IV) The underdeterminedness of an SIPs may preclude the recovery of any ``true'' solution.}
Any solution to (\ref{eq:InvProb}) solves the SIP.
If one is faced with a problem where there is supposed to be a single, true distribution to be recovered, then one will almost certainly not be able to recover it for the following reasons.
As we showed in this work, an infinite number of solutions will exist when $p>q$, as there are infinitely many choices for auxiliary variables. 
Any true distribution comes from joint density of hypothetical observables with the given observables, and this distribution could be very complicated.
Figure \ref{fig:ForwardInverse} shows three potential ``true'' solutions that come from complicated joint distributions for augmented observables $\bd Y_{\text{aug}}$.
Is one choice of auxiliary variables and densities more valid or preferable than others?  In general, no.

Even without direct mention of auxiliary variables, there is no \emph{a priori} reason to prefer BBE, BJW, or any of the possible intuitive solutions.
This is quite different from Bayesian inference where the choice of prior distribution becomes irrelevant as more data is collected.
Moreover, even in the case that $p=q$, there is the risk of infinitely many solutions if the forward map is not uniquely invertible on a given domain, as shown in Section~\ref{ssec:two-to-one} of the appendix.

\textbf{(V) In practice, SIP solutions may not exist.}
Throughout this work we have taken it as given that the SIP solution exists, i.e., that the range $\bd g(\mathcal P)$ contains $\mathcal Q$, and that the system of equations has a solution. 
Underdetermined systems can, in reality, be ``inconsistent'' in that they have none.
Therefore, the existence of a solution needs to be checked and not assumed.
\cite{Moore1977}, for example, gives a test to check for solutions to nonlinear equations within given bounds.

\subsection{Example: CoV Versus Simple Linear Regression}

Suppose that two measurements of a response variable are given $\bd y_{\text{obs}} = (-1,1)^\top$ together with an uncertainty matrix of $\Sigma_y = \sigma^2 \bd I_2$.
An analyst wants to relate these to a predictor variable $x$, having values of $(-1,1)$, through a forward map which is linear in its parameters: $g_x(\bd\theta) = \theta_1 + \theta_2 x  =  [1 \ x] \bd\theta$.

To treat this as an SIP, the analyst might take $\bd Y \sim N(\bd y_{\text{obs}}, \bd\Sigma_y)$, or that the given measurement values were actually population parameters of a Gaussian distribution (I).
The linear forward map is indexed by $x$, but for the problem at hand reduces to 
$\bd Y = \bd X \bd\Theta$ with $\bd X = \begin{bmatrix}  1 & -1  \\  1 & 1  \end{bmatrix}$.
The SIP is solved by 
$f^{\text{CoV}}_{\bd\Theta}(\bd\theta)  =  N_2 \big( \bd X^{-1}\bd y_{\text{obs}}, \ \bd X^{-1} \bd\Sigma_y \bd X^{-\top} \big)$ 
which simplifies to
\begin{align*}
    %
    %
    \bd\Theta
    &\sim  N_2 \left( 
    \begin{bmatrix}  0 \\ 1  \end{bmatrix} , \
    \frac{\sigma^2}{2} \begin{bmatrix}  1 & 0  \\ 0 & 1  \end{bmatrix} \right)  
\end{align*}
When interpolating and extrapolating at some new covariate value $x^*$, the ``predictive'' distribution is a univariate Gaussian derived from a second CoV (II): 
$$ [1 \ x^*] \bd\Theta  \sim  N \left( x^*, \frac{\sigma^2}{2}(1+ x^{*2}) \right) \ .$$
The analyst is thus able to use two datapoints to get the distribution on $\bd\Theta$ and predict (with quantified uncertainty) at any $x^*$ without any checks on model assumptions (because there is no falsifiable model being used!).
If more measurements were obtained for new $x$ values (resulting in $n$ total), the analyst would have to modify their approach to ensure that the number of parameters was greater.
This would be done by 
1] reducing the $n$ measurements to $q = 1$ or 2 observables;
2] expanding the $\bd\Theta$ vector and making the rows of the matrix $\bd X$ correspond to a higher-order polynomial in $x$, ensuring that a new forward map will overfit; or 
3] augmenting the original forward map to include $n$ new parameters $\epsilon_1, \ldots, \epsilon_n$, as per \cite{Butler2018a}.

A statistical analyst would model any number of measurements as $(Y_i | \bd\theta) = \theta_1 + \theta_2 x_i + E_i$ ($i = 1,\ldots, n$), \emph{conditional upon} the covariate values $x_i$ and unknown parameters $\bd\theta$; a marginal, unconditional distribution would \emph{not} be specified for $\bd Y$ as in an SIP.
The right side of the model equation above is a random variable (uppercase $Y_i$) due to the fact that the $E_i$ term (and only this term) on the left side has a distribution, say $N(0, \sigma^2)$.
Each given measurement is a \emph{realization} of a random variable (lowercase $y_i$).
All these data values will be present in the likelihood function, but the argument is $\bd\theta$, not $\bd y$.
A Bayesian might specify a flat prior for $\bd\Theta$ and possibly even take $\sigma^2$ as given to complete the analysis.
In this case it is true that if $n = 2$, the posterior would be the same as the CoV solution, but no reasonable Bayesian would be comfortable with this answer, especially if prediction were required.
Additional measurements beyond $n = 2$ pose no philosophical problem for the statistician, and indeed, the only dimensional consideration on the statistician's mind is ensuring $n >> p = 2$.

\subsection{Example: Calibration of Nuclear Reaction Code}

Here we give an overview of a real type of analysis that can appear to be either a stochastic or statistical IP and is thus useful for pedagogical purposes.

Consider a nuclear reaction model 
in which a collection of ingoing actinide isotopes are subjected to a fission and/or fusion environment controlled by physical fluence parameters and a set of nuclear cross-sections.
The conversion of nuclides within the reactions is governed by a set of differential equations, and the code output is a vector of the resulting actinides at the end of the reaction series.
The total number of fissions is also tabulated by the code.   
In order to compare the outputs to real measurements, per-fission ratios are formed for each resultant actinide (= number of atoms / total fissions, so that scaling factors cancel).
The goal is to estimate the physics parameters (and in the case of the true forensics scenario, ratios of ingoing isotopes) using the observed per-fission quantities.
\begin{figure}[h!t]   \begin{center}
\vspace*{15pt}
\includegraphics[width=\textwidth]{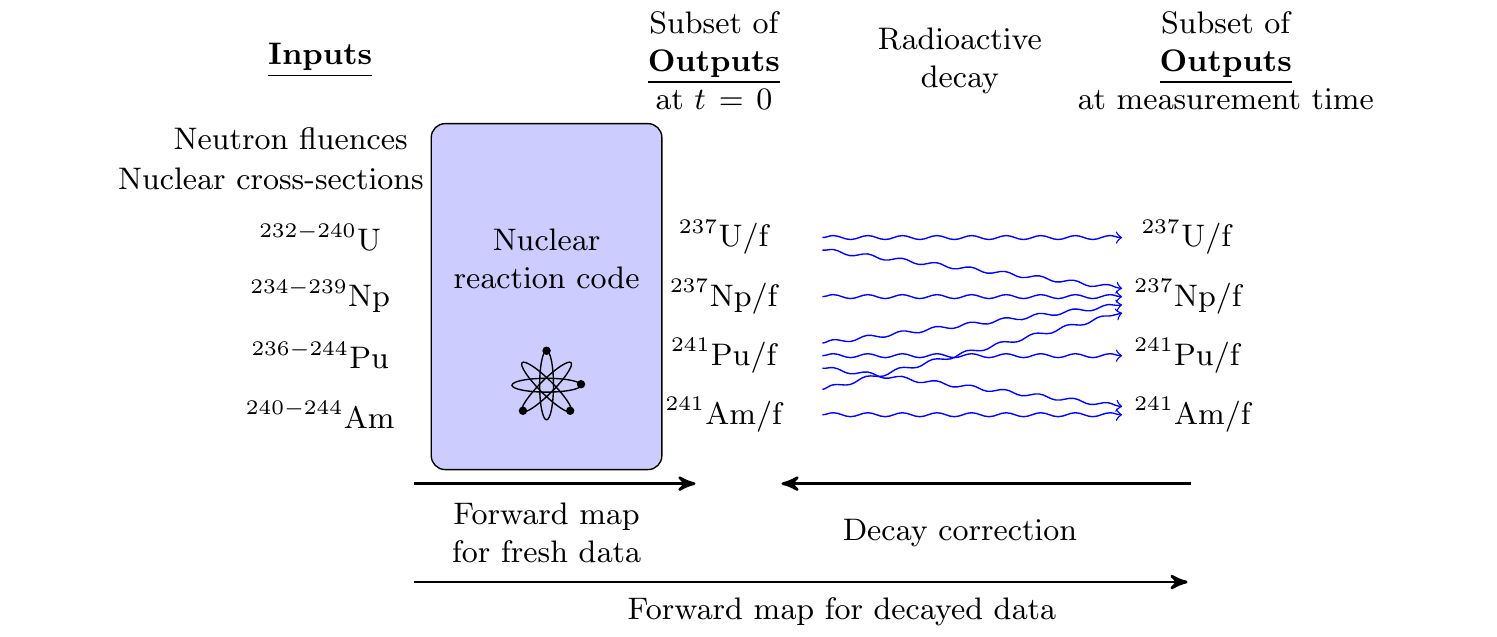} \\
\vspace*{15pt}
\caption{Schematic illustrating the forward problem for a simple scenario of nuclear reactions followed by radioactive decay.}
\label{fig:radio}
\end{center}   \end{figure}

There are multiple actinide per-fission measurements with uncertainties, but these correspond to radiochemical analysis performed some time after the nuclear event.
To account for this time differential, either the Bateman equations (describing radioactive decay) must augment the burn code, or the samples must be decay-corrected back to the end-of-event time ($t=0$) and then used for the calibration.
That is, the forward map for this analysis can either be thought of as the composition of two maps between physics parameters and observables 
$\bd\theta  \stackrel{\text{code}}{\longrightarrow}  \bd Y_0 \stackrel{\text{decay}}{\longrightarrow} \bd Y_t$, 
or solely the code.
Figure \ref{fig:radio} shows this graphically for $\bd Y^\top = ({}^{237}$U$, {}^{237}$Np$, {}^{241}$Pu$, {}^{241}$Am).

Inverting the physics of radioactive decay for a measurement vector is unequivocally a CoV/SIP.
It is therefore tempting to consider the entire calibration process as SIP.
It is especially tempting to do so when, as in many historic radiochemical reports, the original collection of independent decay-corrected samples $\bd y_{0,1}, \ldots, \bd y_{0,n}$ has been reduced to a single vector with uncertainty.   

The most immediate reasons for not treating this nuclear calibration scenario as an SIP are practical.
First, there are often fewer parameters than observables. 
This is almost always true when ingoing isotopics are known, and the goal is to estimate a few fluence parameters.
Treating known masses as unknown to ensure $p > q$ is not appealing.
Second, an SIP solution does not typically exist because the code cannot simultaneously fit all the responses.
After adding discrepancy terms, one is still faced with the possibility of infinitely many solutions since it is not known where the augmented map is one-to-one.

The fundamental reason for not treating this IP as an SIP comes from thinking about what the data represent.
The radiochemical measurements are samples from a population of such quantities, and any reduction to summary values---such as in the historic reports---does not change this fact.
There are unknown population parameters $\bd\theta$, 
and any collection of samples could have been generated by a single $\bd\theta$.
Within the SIP framework, nature's distribution of $\bd\Theta$ would have to change between collections of samples in order to produce their respective variability.
Because this is not how the data is generated, an SIP is not appropriate to infer $\bd\theta$.

\section{Conclusion}   \label{sec:Conclusion}

This paper explored so-called ``stochastic'' inverse problems (SIPs) in great depth.
For the majority of the paper these problems were taken at face value and various solutions explored.
First we provided intuitive solutions derived from a change-of-variables (CoV) wherein the user explicitly controls $p-q$ degrees of freedom. 
The two existing types of solutions in the literature (BBE and BJW) were then shown to be derivable from CoVs. 
We then showed that any solution to an SIP must be directly related to a CoV.
After not questioning the SIP framework, we then gave a lengthy discussion wherein we pointed out a number of fundamental issues inherent to SIPs.
This work thus demonstrates that anyone wanting to treat an estimation or prediction problem as an SIP (as opposed to one of statistical inference) must answer the following questions:
\begin{itemize}
    \item Was the data generated in a way consistent with the SIP framework? How will any possible replicates be treated?
    \item Does the stipulation $p\geq q$ make sense for this problem?
    \item Does an SIP solution exist?
    \item If the answer to the point above is ``yes,'' then: Given that an infinite number of SIP solutions are possible, why is one preferable to any other? On the other hand, given that the SIP solution is almost certainly not the true distribution, will this be problematic for future tasks, e.g., testing, filtering, smoothing, and prediction?
\end{itemize}

\appendix

\section{Example: SIP Involving a Two-to-One Map}  \label{ssec:two-to-one}

As a simple demonstration of the method used in the proof of the Proposition \ref{prop:InfiniteSolutions}, consider the function $g(\theta) = \theta^2$ and a random variable $Y$ on $\mathcal Q = (0,1)$.
If the domain is chosen to be $\mathcal P \defi (-1,1)$, then we can take $\mathcal P_1 \defi (-1,0)$ and $\mathcal P_2 \defi (0,1)$.
The pullback of $f_Y$ through $g_1^{-1}(y) \defi -\sqrt{y}$ and $g_2^{-1}(y) \defi \sqrt{y}$ yields densities $-2\theta f_Y(\theta^2) \mathbbm 1\{\theta \in \mathcal P_1 \}$ and $2\theta f_Y(\theta^2) \mathbbm 1\{\theta \in \mathcal P_2 \}$.
Let $0 < w < 1$ and consider the continuous mixture of the two pullpack densities: 
\begin{align*}
   f^{\text{CoV}}_{\Theta,w}(\theta)  &\defi  2|\theta| f_Y(\theta^2) \Big( w \mathbbm 1\{\theta \in \mathcal P_1 \} + (1-w) \mathbbm 1\{\theta \in \mathcal P_2 \} \Big)  \ .
\end{align*}
By the CoV Theorem, the distribution of $\Gamma \defi \Theta^2$ is 
\begin{align*}
   f_{\Gamma}(\gamma)  &=  
   f^{\text{CoV}}_{\Theta,w} \big( g_1^{-1}(\gamma) \big)  \Big| \first{g_1^{-1}}{\gamma} \Big|  +  
   f^{\text{CoV}}_{\Theta,w} \big( g_2^{-1}(\gamma) \big)  \Big| \first{g_2^{-1}}{\gamma} \Big|  \\
   &=  \frac{1}{2 \sqrt{\gamma}} 
   \left(  f^{\text{CoV}}_{\Theta,w} \big( -\sqrt{\gamma} \big)  +  
   f^{\text{CoV}}_{\Theta,w} \big(  \sqrt{\gamma} \big)  \right)  \\
   &=  f_Y(\gamma)
   \Big( w \mathbbm 1\{-\sqrt{\gamma} \in \mathcal P_1 \} + (1-w) \mathbbm 1\{-\sqrt{\gamma} \in \mathcal P_2 \} \Big)  \\
   & \ \ + f_Y(\gamma)
   \Big( w \mathbbm 1\{\sqrt{\gamma} \in \mathcal P_1 \} + (1-w) \mathbbm 1\{\sqrt{\gamma} \in \mathcal P_2 \} \Big)  \\
   &=  f_Y(\gamma)
   \Big( w \mathbbm 1\{-\sqrt{\gamma} \in \mathcal P_1 \} + (1-w) \mathbbm 1\{\sqrt{\gamma} \in \mathcal P_2 \} \Big)  
   \quad =  f_Y(\gamma)  \ .
\end{align*}
To get a glimpse into the more general case, if the domain had been chosen to be $\mathcal P \defi (-\epsilon,1)$ for $0 < \epsilon < 1$, the domain sets would be $\mathcal P_1 \defi (-\epsilon, 0)$ and $\mathcal P_2 \defi (0, \epsilon)$, but now $\mathcal P_3$ would be the set $(\epsilon, 1)$ where the function is one-to-one.
The infinite set of solutions would take the form
\begin{align*}
   f^{\text{CoV}}_{\Theta,w}(\theta)  &\stackrel{\text{def}}{\propto}  2|\theta| f_Y(\theta^2) \cdot
   \left\{ 
   \begin{array}{cc}
   w \mathbbm 1\{\theta < 0 \} + (1-w) \mathbbm 1\{\theta > 0 \}   &  -\epsilon < \theta < \epsilon   \\
   1   &  \phantom{-} \epsilon < \theta < 1
   \end{array}  \right.  \ .
\end{align*}
Using $Y \sim $ $Unif(0,1)$ with $\epsilon = 0.5$, two solutions are shown in \ref{fig:CoV1}.  
\begin{figure}[h!t]   \begin{center}
\includegraphics[width=0.5\textwidth]{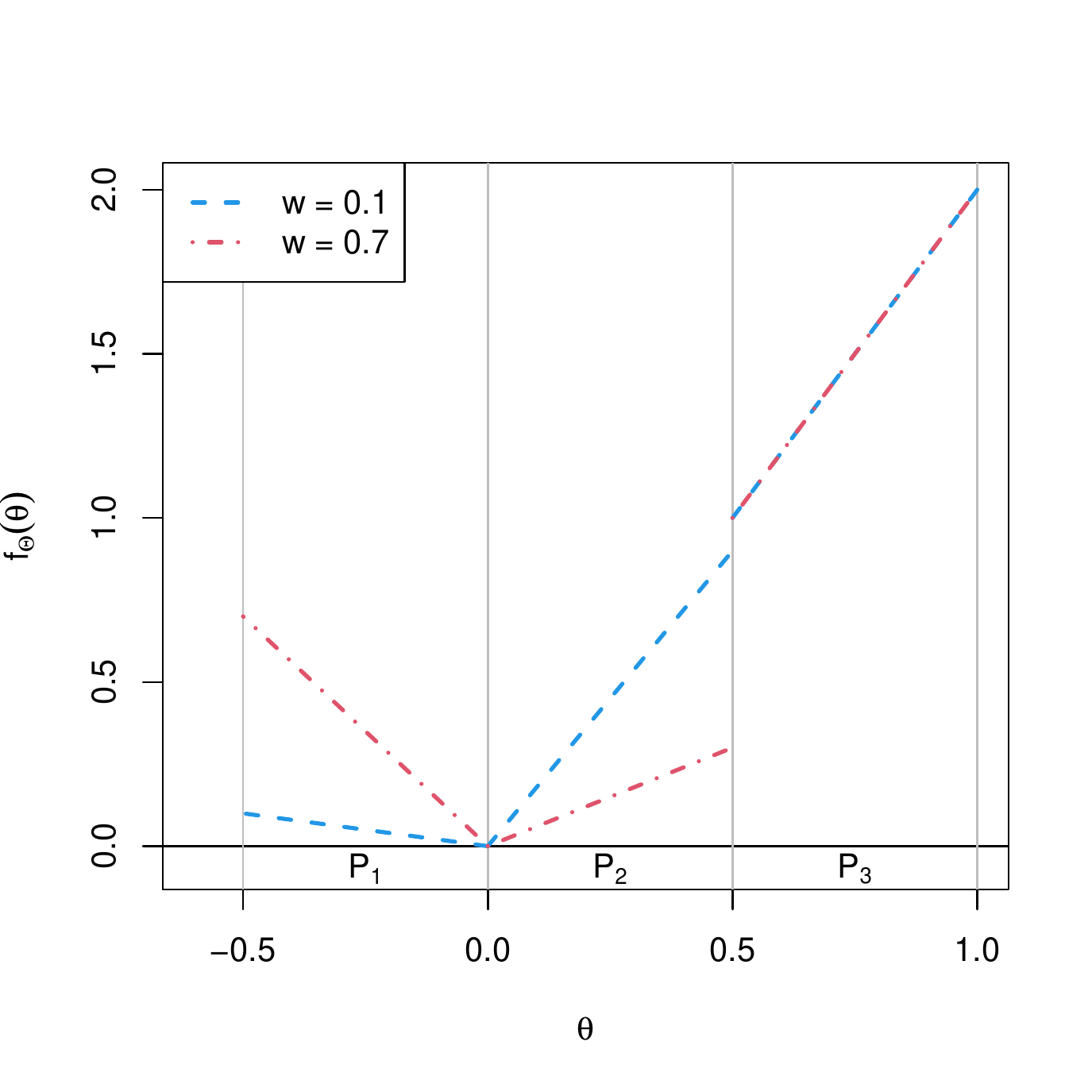} 
\caption{
Two of the infinitely many solutions when $Y \sim $Unif$(0,1)$ and $g(\theta) = \theta^2$ on $(-0.5, 1)$ determined by the choice of $w$.
}
\label{fig:CoV1}
\end{center}   \end{figure}


\section{Additional Derivations}

\subsection{BJW Solution: Linear Transformation of a Multivariate Gaussian Vector}  \label{sec:App1}

Here we verify the claim in Section~\ref{sec:ExampleLinearBJW} that the proposed Gaussian distribution is indeed the BJW solution under the linear map $\bd g(\bd\theta) = \bd A \bd\theta$.

\begin{lemma}
For $\bd\Theta \sim f^{\text{BJW}} \equiv N_p(\wtilde{\bd\mu}, \wtilde{\bd\Sigma})$ with moments defined by \ref{eq:muti}, \ref{eq:sigti}, $\bd A \bd\Theta \sim N_q(\bd\mu_y, \bd\Sigma_y)$, confirming that the proposed BJW solution indeed solves the SIP (\ref{eq:InvProb}).   \label{lem:BJW}
\end{lemma}
\begin{proof}
We need to show that $\bd A \wtilde{\bd\mu} = \bd\mu_y$ and $\bd A \wtilde{\bd\Sigma} \bd A^\top = \bd\Sigma_y$.
We will show the second claim about covariances and then use it to show the first claim about the means.

Starting from \ref{eq:sigti},
\[
\wtilde{\bd\Sigma}^{-1}  =  \bd\Sigma_{\theta}^{-1} + \bd A^T \big\{ \bd\Sigma_y^{-1} - (\bd A\bd\Sigma_\theta \bd A^T)^{-1} \big\} \bd A  \ .
\]
Let $\bd D=\bd\Sigma_\theta^{-1}$, $\bd U = \bd A^T$, $\bd E = \bd \Sigma_y^{-1} - (\bd A \bd\Sigma_\theta \bd A^T)^{-1}$, and $\bd V = \bd A$. 
The  Woodbury formula states $(\bd D + \bd U \bd E \bd V)^{-1} = \bd D^{-1} - \bd D^{-1} \bd U (\bd E^{-1} + \bd V\bd D^{-1}\bd U)^{-1}\bd V\bd D^{-1}$, so that
\begin{align*}
\wtilde{\bd\Sigma}  =  
\bd\Sigma_\theta - \bd\Sigma_{\theta} \bd A^T\left\{(\bd\Sigma_y^{-1} - (\bd A\bd\Sigma_\theta \bd A^T)^{-1})^{-1} + \bd A \bd \Sigma_\theta \bd A^T \right\} ^{-1} \bd A \bd\Sigma_\theta  \ .
\end{align*}
Now consider $\bd A \wtilde{\bd\Sigma} \bd A^T$, and let $\bd B = \bd A \bd\Sigma_\theta \bd A^T$. 
Then,
\[ \bd A \wtilde{\bd\Sigma}\bd A^T = \bd B - \bd B \left\{(\bd\Sigma_y^{-1} - \bd B^{-1}) ^{-1} + \bd B \right\}^{-1} \bd B  \ . \]
Now redefine $\bd D = (\bd\Sigma_y^{-1} - \bd B^{-1})^{-1}$ and $\bd E = \bd B$ (with $\bd U = \bd V = \bd I$), and again apply the Woodbury formula.  
Then $(\bd D + \bd E)^{-1} = \bd B^{-1} - \bd B^{-1}\bd\Sigma_y \bd B^{-1}$.
Thus,
\begin{align*}
    \bd A \wtilde{\bd\Sigma}\bd A^T &= \bd B - \bd B(\bd B^{-1} - \bd B^{-1} \bd\Sigma_y \bd B^{-1}) \bd B  \\
    %
    %
    &=\bd\Sigma_y  \ . 
\end{align*}

Next, to see the first claim about the mean vectors, start from \ref{eq:muti} and add and subtract the term $\bd A^T \bd\Sigma_y^{-1}\bd A\bd\mu_\theta$ within the parentheses:
\begin{align}
\wtilde{\bd\mu} &= \wtilde{\bd\Sigma} \left( \bd A^T \bd\Sigma_y^{-1}\bd\mu_y -\bd A^T \bd\Sigma_y^{-1}\bd A\bd\mu_\theta  +  
\big\{ \bd A^T \bd\Sigma_y^{-1}\bd A\bd  -  \bd A^T (\bd A \bd\Sigma_\theta \bd A^T)^{-1} \bd A \big\} \bd\mu_\theta  +  
\bd\Sigma^{-1}_\theta \bd\mu_\theta \right)  \nonumber \\
&= \wtilde{\bd\Sigma} \left( \bd A^T \bd\Sigma_y^{-1}\bd\mu_y -\bd A^T \bd\Sigma_y^{-1}\bd A\bd\mu_\theta + \big\{ \wtilde{\bd\Sigma}^{-1} - \bd\Sigma^{-1}_\theta \big\} \bd\mu_\theta  + 
\bd\Sigma^{-1}_\theta \bd\mu_\theta \right)  \nonumber \\
&= \wtilde{\bd\Sigma} \left( \bd A^T \bd \Sigma_y^{-1}\bd\mu_y -\bd A^T \bd\Sigma_y^{-1}\bd A\bd\mu_\theta \right) + \bd\mu_\theta  \nonumber \\
&= \wtilde{\bd\Sigma}  \bd A^T \bd\Sigma_y^{-1}(\bd \mu_y -\bd A\bd\mu_\theta) + \bd\mu_\theta  \nonumber  \ .
\end{align}
Now consider $\bd A\wtilde{\bd\mu}$ and use the fact above that $\bd A \wtilde{\bd\Sigma} \bd A^T = \bd\Sigma_y$ so that 
\begin{align*}
\bd A\wtilde{\bd\mu}  &=  \big( \bd A \wtilde{\bd\Sigma} \bd A^T \big) \bd\Sigma_y^{-1}(\bd\mu_y -\bd A \bd\mu_\theta) + \bd A\bd\mu_\theta  \\
 &= \bd\mu_y  \ .
\end{align*}
\end{proof}

\subsection{BJW Solution: Gaussian Mean Estimation as ``Stochastic Map'' Inversion}  \label{sec:App2}

The forward map is $\bd g(\bd\theta) \defi \bd A \bd\theta$ where $\bd A \defi [\bd 1_n \vdots \bd I_n]$ and 
\\ $\bd\theta_n \defi (\mu, \epsilon_1, \ldots, \epsilon_n)^\top$.
The observable distribution is $N_n(\bd\mu_y, \bd\Sigma_y)$ with $\bd\Sigma_y = \sigma^2_y \bd I_n$.
The initial distribution is $\ut f_{\bd\Theta}  \equiv  N_{n+1}(\bd\mu_\theta, \bd\Sigma_\theta)$ with $\bd\mu_\theta^\top  =  (\mu_0, \bd 0_n^\top)$ and 
\\ $\bd\Sigma_\theta  =  \text{diag}(\sigma_0^2, \sigma^2_\epsilon, \ldots, \sigma^2_\epsilon)$.
The BJW solution in this case is $N_{n+1}(\wtilde{\bd\mu}, \wtilde{\bd\Sigma})$ with mean and covariance given by \ref{eq:muti}, \ref{eq:sigti}, or equivalently,
\begin{align*}
   \wtilde{\bd\Sigma}^{-1}
   \wtilde{\bd\mu}  &=  \bd A^\top \bd\Sigma_y^{-1} \bd\mu_y  
               -  \bd A^\top (\bd A \bd\Sigma_\theta \bd A^\top)^{-1} \bd A \bd\mu_\theta  
               +  \bd\Sigma_\theta^{-1} \bd\mu_\theta   \\
   \wtilde{\bd\Sigma}^{-1}  &=  \bd A^\top \bd\Sigma_y^{-1} \bd A  
                  -  \bd A^\top (\bd A \bd\Sigma_\theta \bd A^\top)^{-1} \bd A  
                  +  \bd\Sigma_\theta^{-1}  \ . 
\end{align*}
For any $\sigma^2$ term, we will take $\tau = 1 / \sigma^2$ to be the corresponding precision.  Also define $\bd J_n = \bd 1_n \bd 1_n^\top$.
The most complicated term in the equations above is
\begin{align*}
   \bd A^\top (\bd A \bd\Sigma_\theta \bd A^\top)^{-1} \bd A  
   &=  
   \begin{bmatrix}  n c_2  &  c_2 \bd 1_n^\top  \\  c_2 \bd 1_n  &  \tau_\epsilon \bd I_n - c_1 \bd J_n  \end{bmatrix}  \\
   c_1  &\defi  \frac{\tau^2_\epsilon}{n\tau_\epsilon + \tau_0}  
   \quad\quad
   c_2  \defi  \tau_\epsilon - n c_1  =  \frac{\tau_\epsilon \tau_0 }{n\tau_\epsilon + \tau_0}  \ .
\end{align*}
Now the precision matrix is
\begin{align*}
   \wtilde{\bd\Sigma}^{-1}  
   &=
   \begin{bmatrix}  n\tau_y  &  \tau_y \bd 1_n^\top  \\  \tau_y \bd 1_n  &  \tau_y \bd I_n \end{bmatrix}  -
   \begin{bmatrix}  n c_2  &  c_2 \bd 1_n^\top  \\  c_2 \bd 1_n  &  \tau_\epsilon \bd I_n - c_1 \bd J_n  \end{bmatrix}  +
      \begin{bmatrix}  \tau_0  &  \bd 0_n^\top  \\  \bd 0_n  &  \tau_\epsilon \bd I_n  \end{bmatrix}  \\
   &=
   \begin{bmatrix}  n c_3 + \tau_0  &  c_3 \bd 1_n^\top  \\  c_3 \bd 1_n  &  \tau_y \bd I_n + c_1 \bd J_n  \end{bmatrix}  
   \quad\quad 
   c_3  \defi  \tau_y - c_2  \ .
\end{align*}
The covariance is 
\begin{align*}
   \wtilde{\bd\Sigma}  
   &=  
   \begin{bmatrix}  a  &  \bd b^\top  \\  \bd b  &  \bd C  \end{bmatrix}^{-1}
   =  
   \begin{bmatrix}  \frac{1}{a} + \frac{1}{a^2} \bd b^\top (2,2) \bd b  &  -\frac{1}{a} \bd b^\top (2,2)  \\  -\frac{1}{a} (2,2) \bd b  &  (2,2)  \end{bmatrix}  \\
   (2,2)^{-1}  &=  \bd C - \frac{1}{a} \bd b \bd b^\top  
   =  \tau_y \bd I_n  +  c_4 \bd J_n  \\
   (2,2)  &=  \sigma^2_y \bd I_n - c_5 \bd J_n  \\
   c_4  &\defi  \frac{\tau_\epsilon^2}{n \tau_\epsilon + \tau_0}  -  \frac{c_3^2}{n c_3 + \tau_0}  
   \quad\quad 
   c_5  \defi  \frac{\sigma^4_y c_4}{n \sigma^2_y c_4 + 1}  \ .  
\end{align*}
and hence
\begin{align}
   \wtilde{\bd\Sigma}  
   &=  
   \begin{bmatrix}  c_7  &  c_6\bd 1_n^\top  \\  c_6 \bd 1_n  &  \sigma^2_y \bd I_n - c_5 \bd J_n  \end{bmatrix}  
   \quad =  \begin{bmatrix}  \mathcal O(n^{-1})  &  -\mathcal O(n^{-1}) \bd 1_n^\top  \\  -\mathcal O(n^{-1}) \bd 1_n  &  \sigma^2_y \bd I_n - \mathcal O(n^{-1}) \bd J_n  \end{bmatrix}  \\
   c_6  &\defi  - \frac{ c_3(\sigma^2_y - n c_5) }{n c_3 + \tau_0}  
   \quad\quad 
   c_7  \defi  \frac{1}{n c_3 + \tau_0}  +  \frac{n c_3^2 (\sigma^2_y - n c_5)}{(n c_3 + \tau_0)^2}  \ .  \nonumber
\end{align}
%

Returning to the mean, if $\overline\mu_y  \defi  n^{-1} \sum_{i = 1}^n \mu_i$, then
\begin{align*}
   \wtilde{\bd\Sigma}^{-1}
   \wtilde{\bd\mu}  
   &=
   \begin{bmatrix}  \tau_y \bd 1_n^\top \bd\mu_y - n c_2\mu_0 + \tau_0 \mu_0  \\  \tau_y \bd\mu_y - c_2 \mu_0 \bd 1_n  \end{bmatrix}  \\
   %
   \wtilde{\bd\mu}  
   &=  
   \begin{bmatrix}  c_7  &  c_6\bd 1_n^\top  \\  c_6 \bd 1_n  &  \sigma^2_y \bd I_n - c_5 \bd 1_n \bd 1_n^\top  \end{bmatrix} 
   \begin{bmatrix}  n \tau_y \overline\mu_y + (\tau_0 - n c_2) \mu_0  \\  \tau_y \bd\mu_y - c_2 \mu_0 \bd 1_n  \end{bmatrix}  \\
   &=  
   \begin{bmatrix}  
   (n c_6 + n c_7) \tau_y \overline\mu_y  +  \left( -n c_2 c_6 + c_7 (\tau_0 - n c_2) \right)\mu_0  \\  
   \bd\mu_y  +  (n c_6 - n c_5) \tau_y \overline\mu_y \bd 1_n  +  \left( -c_2 \sigma^2_y + n c_2 c_5 + c_6 (\tau_0 - n c_2)  \right) \mu_0 \bd 1_n
   \end{bmatrix}  \\
   &=  
   \begin{bmatrix}  \left( \sigma^2_y - \mathcal O(n^{-1}) \right) \tau_y \overline\mu_y  +  \mathcal O(n^{-1}) \mu_0  \\  
   \bd\mu_y  +  \left( -\sigma^2_y + \mathcal O(n^{-1}) \right) \tau_y \overline\mu_y \bd 1_n  -  \mathcal O(n^{-1}) \mu_0 \bd 1_n \end{bmatrix}  \ ,  
 \end{align*}
 and hence
 \begin{align}
   \wtilde{\bd\mu}  
   &=
   \begin{bmatrix}  \overline\mu_y + \mathcal O(n^{-1})  \\  (\bd\mu_y - \overline\mu_y \bd 1_n) +  \mathcal O(n^{-1}) \bd 1_n  \end{bmatrix}  \ . 
\end{align}
%



\scriptsize

\bibliographystyle{asa}
\bibliography{Change_of_Vars.bib}


\end{document}